\documentclass[sigconf,authorversion,nonacm]{acmart}
\usepackage[T1]{fontenc}
\usepackage[utf8]{inputenc}
\usepackage{graphicx,tikz,xcolor,hyperref,booktabs,upquote,csquotes,textcomp,subfig,multirow,caption}
\usepackage{pifont}
\usepackage{xcolor}
\usepackage{xspace}
\usepackage[ruled,vlined,longend,linesnumbered]{algorithm2e}
\usepackage{cleveref}
\usepackage{paralist,array}
\usepackage[export]{adjustbox}

\usetikzlibrary{arrows,decorations.pathmorphing,backgrounds,positioning,fadings,shadings,through,fit,petri,calc}

\newcommand{\PreserveBackslash}[1]{\let\temp=\\#1\let\\=\temp}
\newcolumntype{C}[1]{>{\PreserveBackslash\centering}p{#1}}

\crefname{algocf}{alg.}{algs.}
\Crefname{algocf}{Algorithm}{Algorithms}
\def\name{{\sc TAP}\xspace}
\newcommand{\node}{i}
\newcommand{\prefixkey}{\chi_{s}}
\newcommand{\prefixval}{\phi_{s}}

\newcommand{\bparagraph}[1]{\textbf{#1}.\ }

\newcommand{\gitrepo}{\url{https://github.com/tap-group/transparent-data-service}}

\begin{document}

\title{TAP: Transparent and Privacy-Preserving Data Services}

\author{Dani\"el Reijsbergen,$^\dagger$ Aung Maw,$^\dagger$ Zheng Yang,$^\ddagger$ Tien Tuan Anh Dinh$,^\dagger$ and Jianying Zhou$^\dagger$\\[0.1cm]}
\affiliation{$^\dagger$Singapore University of Technology and Design, Singapore, Singapore \\ $^\ddagger$Southwest University, Chongqing, China \\[0.2cm]}

\begin{abstract}
Users today expect more security from services that handle their data. In addition to traditional data
privacy and integrity requirements, they expect {\em transparency}, i.e., that the service's processing of the data is
verifiable by users and trusted auditors. Our goal is to build a multi-user system that provides data
privacy, integrity, and transparency for a large number of operations, while achieving practical
performance.  

To this end, we first identify the limitations of existing approaches that use \textit{authenticated data structures}. 
 We find that they fall into two categories: 1) those that hide each user's data
from other users, but have a limited range of verifiable operations (e.g., CONIKS, Merkle$^2$, and Proofs of Liabilities), and 2) those
that support a wide range of verifiable operations, but make all data publicly visible (e.g., IntegriDB and
FalconDB). We then present \name to address the above limitations. The key component of \name is a novel
tree data structure that supports efficient result verification, and relies on independent audits that use
zero-knowledge range proofs to show that the tree is constructed correctly without revealing user data.  \name
supports a broad range of verifiable operations, including quantiles and sample standard deviations. We conduct a comprehensive
evaluation of \name, and compare it against two state-of-the-art baselines, namely IntegriDB and Merkle$^2$,
showing that the system is practical at scale. 

\end{abstract}

\maketitle

\newcommand{\rand}{\stackrel{{\scriptscriptstyle\$}}{\leftarrow}} 

\newcommand{\negl}{\ensuremath{\mathsf{negl}(\kappa)}}
\newcommand{\GameP}{\ensuremath{\mathsf{Game}}}
\newcommand{\Adv}{\ensuremath{\mathsf{Adv}}}
\newcommand{\Adversary}{\ensuremath{\mathcal{A}}}

\newcommand{\ParmsC}{\ensuremath{\text{P}_\text{c}}}
\newcommand{\VSpace}{\ensuremath{\mathcal{V}_\text{c}}}
\newcommand{\RSpace}{\ensuremath{\mathcal{R}_\text{c}}}
\newcommand{\CSpace}{\ensuremath{\mathcal{C}_\text{c}}}

\renewcommand{\P}{\mathbb{P}}
\newcommand{\RR}{\mathcal{R}}
\newcommand{\DD}{\mathcal{D}}
\newcommand{\sens}{\Delta}
\newcommand{\query}{Q}
\newcommand{\res}{R}
\newcommand{\trueres}{R^*}
\newcommand{\data}{D}
\newcommand{\datb}{D'}
\newcommand{\noisea}{Z}
\newcommand{\noiseb}{Z'}
\newcommand{\pp}{p'}
\newcommand{\BB}{\mathcal{B}}
\newcommand{\bound}{b}

\newcommand{\ParmsTAP}{\ensuremath{\text{P}_\text{tap}}}
\newcommand{\TAPInitialize}{\ensuremath{\mathsf{Initialize}}}
\newcommand{\TAPEpochSecretGen}{\ensuremath{\mathsf{EpochSecretGen}}}

\newcommand{\TAPQuery}{\ensuremath{\mathsf{Query}}}
\newcommand{\TAPVerify}{\ensuremath{\mathsf{Verify}}}
\newcommand{\TAPCheckEpoch}{\ensuremath{\mathsf{EpochCheck}}}

\newcommand{\TAPDataSt}{\ensuremath{\text{DS}}}
\newcommand{\TAPProof}[1]{\ensuremath{\pi_{#1}}}
\newcommand{\TAPQueryMsg}[1]{\ensuremath{\text{M}_{#1}}}
\newcommand{\TAPQueryRestult}[1]{\ensuremath{\text{R}_{#1}}}

\newcommand{\TAPProofSpace}{\ensuremath{\mathcal{PF}_{\text{tap}}}}
\newcommand{\TAPSKSpace}{\ensuremath{\mathcal{SK}_{\text{tap}}}}
\newcommand{\TAPSSSpace}{\ensuremath{\mathcal{SS}_{\text{tap}}}}
\newcommand{\TAPVSpace}{\ensuremath{\mathcal{V}_{\text{tap}}}}
\newcommand{\TAPRSpace}{\ensuremath{\mathcal{R}_{\text{tap}}}}

\newcommand{\TAPMSpace}{\ensuremath{\mathcal{M}_{\text{tap}}}}
\newcommand{\TAPVTSpace}{\ensuremath{\mathcal{VT}_{\text{tap}}}}

\newcommand{\TAP}{\ensuremath{\mathsf{TAP}}}
\newcommand{\Privacy}{\ensuremath{\mathsf{Priv}}}
\newcommand{\Transparancy}{\ensuremath{\mathsf{Trans}}}

\newcommand{\QueryType}{\ensuremath{\mathsf{QType}}}
\newcommand{\Qinsert}{\ensuremath{\mathsf{insert}}}
\newcommand{\Qlookup}{\ensuremath{\mathsf{lookup}}}
\newcommand{\Qsum}{\ensuremath{\mathsf{sum}}}
\newcommand{\Qcount}{\ensuremath{\mathsf{count}}}
\newcommand{\Qaverage}{\ensuremath{\mathsf{average}}}
\newcommand{\Qminmax}{\ensuremath{\mathsf{min\text{-}max}}}
\newcommand{\Qquantile}{\ensuremath{\mathsf{quantile}}}

\newcommand{\OracleTAPQuery}{\ensuremath{\mathcal{O}_{\mathsf{MQ}}}}
\newcommand{\OracleInsertH}{\ensuremath{\mathcal{O}_{\mathsf{IH}}}}
\newcommand{\OracleCorrupt}{\ensuremath{\mathcal{O}_{\mathsf{RV}}}}

\newcommand{\CorruptList}{\ensuremath{\mathsf{CL}}}

\newcommand{\PPTPRelation}{\ensuremath{\mathcal{RL}}}
\newcommand{\Vmax}{\ensuremath{\mathsf{v}_{\text{max}}}}

\newcommand{\CN}{\ensuremath{\text{cnt}}}
\newcommand{\DN}{\ensuremath{\mathsf{DN}}}

\newcommand{\Qlist}{\ensuremath{\mathsf{QL}}}
\newcommand{\Hlist}{\ensuremath{\mathsf{HL}}}

\newcommand{\Qcouter}{\ensuremath{\text{f}}}
\newcommand{\Ccounter}{\ensuremath{\mathsf{g}}}
\newcommand{\Icounter}{\ensuremath{\mathsf{l}}}
\newcommand{\PrivacyType}{\ensuremath{\mathsf{PT}}}

\newcommand{\OracleVChallenge}{\ensuremath{\mathsf{PrivChallenge}}}
\newcommand{\OracleMChallenge}{\ensuremath{\mathsf{MinMaxChallenge}}}
\newcommand{\ValuePrivacy}{\ensuremath{\mathsf{value\text{-}priv}}}
\newcommand{\MMPrivacy}{\ensuremath{\mathsf{minmax\text{-}priv}}}
\newcommand{\MsgSUM}{\ensuremath{\mathsf{MSUM}}}

\section{Introduction}
\label{sec:introduction}

Many of today's applications collect large amounts of user data and make decisions that have a direct impact on the user.
One example is a utility company that collects power
usage data from users and charges them different rates based on peak/off-peak periods per region. 
Another example is a road pricing system that determines real-time traffic conditions based on the number of cars on the
road, and charges motorists appropriate congestion fees. 
The third example is
an advertising service that monitors clicks and pays the publisher for displaying the advertisement based on 
the click-through rate. One desirable property of the applications above is {\em transparency} \cite{chen2008access,pang2005verifying}, which allows
users to verify that the computation done on their data has been executed correctly. This property is stronger than
simply ensuring data \textit{integrity}, as it protects users from malicious or
compromised service providers who ignore data or tamper with the computation.

\begin{figure}[t]
\centering
\includegraphics[width=0.9\linewidth]{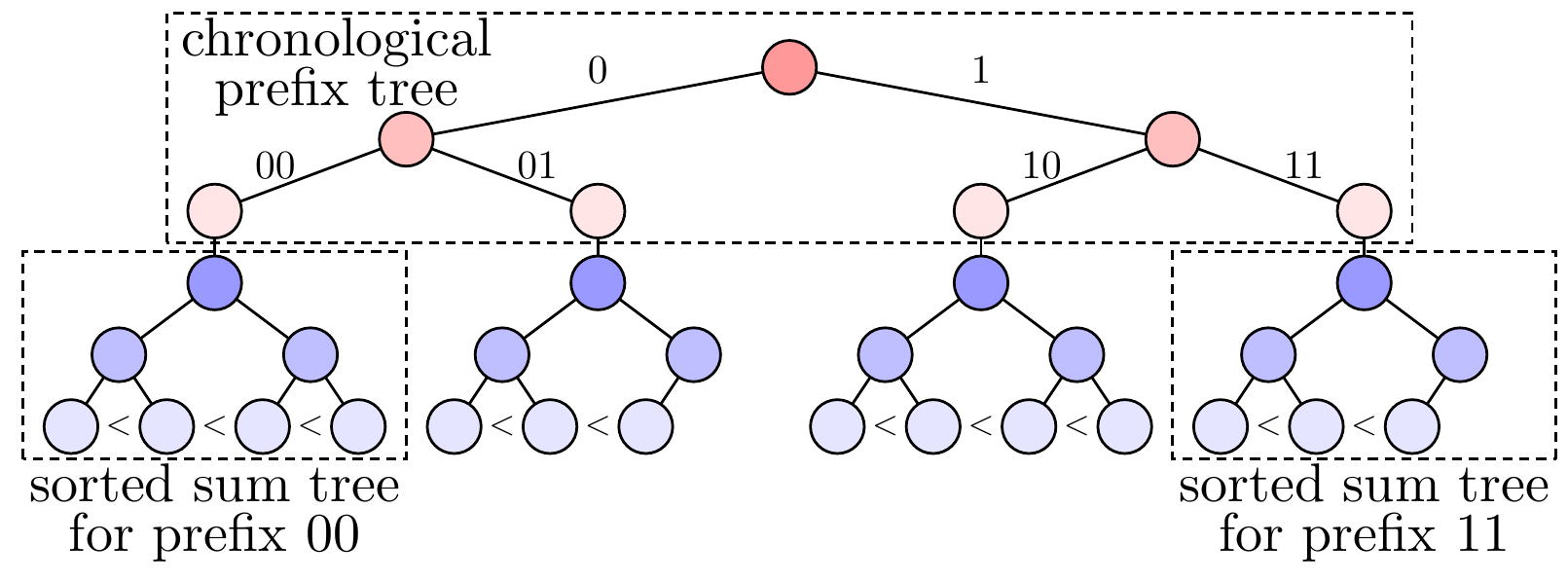}
\caption{\name's authenticated data structure.} 
\label{fig:solution_ads} 
\end{figure}

One approach toward transparent data services is for the provider to make the raw data available to users.
For example, the provider can use a public bulletin board to store the data. Users can
then perform computations directly on the data,
 but they
have no data \textit{privacy}. This is unacceptable in applications such as dynamic road pricing, as it would allow any user to track the movements of other users using the raw data. Alternatively, the service provider could store the data and return only query results. For example, an advertising service may have an API that returns the total number of clicks on an advertisement over a given period. However, this approach cannot guarantee transparency, which is troublesome when the provider has a financial incentive to falsify query results.

Our goal is to build a data service system with the following properties. First, it supports a wide range of queries that compute \textit{aggregates} over the data of many independent users.
Second, it protects data privacy from adversarial users. Third, it protects data integrity and transparency from an
adversarial provider. Finally, it has a reasonable
\textit{performance} overhead, and scales well when the number of users increases.

One important primitive for realizing our goal is the \textit{Authenticated Data
Structure}~(ADS)~\cite{tamassia2003authenticated}, which allows users to detect incorrect
results returned by the provider. 
However, all ADS proposals in the literature fall into two broad categories, which both
fall short in meeting our requirements. 
The first category consists of approaches that guarantee privacy, but support a limited set of queries. This category includes \textit{transparency logs} such as
certificate transparency \cite{laurie2014certificate} and its variants such as revocation transparency
\cite{laurie2012revocation}, extended certificate transparency \cite{ryan2014enhanced}, Trillian
\cite{eijdenberg2015verifiable}, CONIKS \cite{melara2015coniks}, and Merkle$^2$ \cite{hu2021merkle}.
Transparency logs store data in the form of key-value tuples and allow for public auditing. 
However, transparency logs only support data insertion, removal, and look-up
operations, and thus fail to meet our first requirement. Also included in the first category are Proofs of Liabilities (PoLs) \cite{dagher2015provisions,ji2021generalized}, which allow users to prove that certain sums of user values are non-negative without revealing the underlying values.
The second category consists of SQL-based
authenticated databases such as IntegriDB \cite{zhang2015integridb} and FalconDB \cite{peng2020falcondb}.
These databases allow users to verify the results of a wide range of SQL queries -- i.e., sum, count, average, min, and max.
However, they are either designed for a single user~\cite{zhang2015integridb}, thus failing our first
requirement, or make all insert queries public \cite{peng2020falcondb}, thus violating the second
requirement.  Other approaches exist that only support transparent range queries \cite{chen2008access,pang2005verifying,poddar2016arx} 
and hence fail the first requirement.

In this work, we present a \underline{T}ransparent \underline{a}nd \underline{P}rivacy-Preserving ADS named \textit{TAP}, a data service system that meets our four design goals. It addresses the limitations of
existing systems with a novel tree data structure, depicted in \Cref{fig:solution_ads}, that combines the features of existing approaches that are best suited to our context. The data structure consists of a chronological prefix tree (as in Merkle$^2$),
and each leaf in the
prefix tree is the root of a Merkle sum tree (as in PoLs) that is sorted (as in IntegriDB).  \name adopts the same system model as
CONIKS and Certificate Transparency, in which individual users \textit{monitor} the inclusion of their data, and there
exist some \textit{auditors} that verify the relevant properties -- e.g., ordered and append-only -- of the data
structure. The prefix trees enable efficient monitoring and auditing, similar to the data structures of CONIKS and Merkle$^2$. Meanwhile, the sorted sum trees enable fast verification for a wide range of operations, including sum, count, average,
min, max, and sample standard deviation queries. \name also supports the quantile query~\cite{li2010authenticated} that allows for the computation of fine-grained statistics, e.g., the median or the 5th percentile on sliding
windows.

\name is designed to store data from multiple users, thus meeting the first requirement. It protects data privacy -- the second requirement -- by storing cryptographic commitments instead of raw data, and by publishing zero-knowledge
proofs. \name's Merkle tree structure ensures data integrity and allows users to verify the correctness of a
broad range of queries -- the third requirement -- by generating Merkle proofs and
zero-knowledge range proofs for the commitments' underlying values. In our setting -- i.e., a single data table in which aggregates are computed over sliding windows -- \name has better performance
than previous systems, because it maintains one single tree, as opposed to the many trees in IntegriDB, and the
tree is smaller than the tree in Merkle$^2$. 
The computation and bandwidth overheads of \name are linear in the size of the sliding windows, but logarithmic
in the size of the entire tree. 

To reason about the security of \name, we require the following: that each user adds one data entry per time slot, that the set of users (but not their data) at each time slot is known to a super-auditor (e.g., a regulator or watchdog), and that the fraction of adversarial users is bounded.
We present a detailed analysis of the properties of TAP in this setting, and find an explicit tradeoff between transparency and privacy. We focus on guaranteeing perfect transparency at the cost of revealing query results, which cannot be trivially linked to user identities. 
In \Cref{sec:differential_privacy}, we discuss a method that 
guarantees $(\epsilon,\delta)$-differential privacy \cite{dwork2008differential}. 
Our final contribution is a full, publicly accessible implementation of \name, and we conduct a broad range of experiments to evaluate its performance. 
 We
compare \name against two relevant baselines -- Merkle$^2$ and
IntegriDB. The results show that the system has reasonable overhead, and that it outperforms the
baselines in many cases.  

\bparagraph{Contributions} We make the following contributions:
\begin{enumerate}
\setlength{\itemsep}{0pt}
\item We present a survey of existing ADS approaches, and discuss their
limitations in today's emerging applications. 

\item We present \name, a multi-user data service whose ADS combines elements from CONIKS, Merkle$^2$, IntegriDB, and PoLs to protect data privacy and integrity, while
providing transparency to a wide range of operations.

\item We formally analyze \name and prove that it only reveals the results of queries. 

\item We present a full implementation of \name and evaluate its performance.
We compare it against two baselines, namely IntegriDB and Merkle$^2$, and show that \name outperforms the
baselines in many cases.
\end{enumerate}

\bparagraph{Outline} The remainder of this work is organized as follows. Section~\ref{sec:system_model}
describes the system model, use case examples, threat model, and our design goals. Section~\ref{sec:survey}
discusses related systems built on top of ADSs. Section~\ref{sec:solution} presents
\name. Section~\ref{sec:analysis} provides security and performance analysis of \name, and discusses
its current limitations. Section~\ref{sec:evaluation} contains the detailed performance evaluation and
comparison against two state-of-the-art baselines. We discuss the practical aspects of \name in Section~\ref{sec:discussion}, and Section~\ref{sec:conclusions} concludes the paper.

\section{Model \& Requirements}
\label{sec:system_model}

In this section, we first discuss the general system and data models. Next, we present three use cases for transparent data services and discuss how they fit into our models. Finally, we present the threat model and system
requirements. 

\subsection{Model Entities}
\label{sec:model_entities}

Our system model
consists of the following types of entities:

\bparagraph{Users} Users send data to the server and issue \textit{queries} on the aggregate data through a client. Each user \textit{monitors} the data structure by verifying
that her data is properly stored by the server and verifies that query results are computed correctly. In practice, monitoring can be automated, e.g., the app on the user's mobile phone that shows bills or payments can verify the displayed values by querying the server's ADS.

\bparagraph{Server} The server stores the data provided by the users in a database, and maintains an
ADS on top of the data. It computes \textit{responses} to user queries, and generates \textit{proofs} for the
responses using the ADS.  

\bparagraph{Auditors} Auditors validate the server's ADS. In particular, they verify
that it is \textit{append-only}, i.e., data is never modified or deleted, and that certain data has been \textit{sorted} correctly. We will discuss these checks in more detail in later sections. 

\bparagraph{Bulletin board} The server periodically publishes the {\em digest} of its ADS to an immutable bulletin board, e.g., a public blockchain. 
Users and auditors download the latest digests during monitoring, auditing, and query verification. The only goal of the bulletin board is to prevent equivocation -- i.e., the server presenting different versions of the ADS to different entities -- and can therefore be replaced by a {gossip protocol} that would serve the same purpose.

\Cref{fig:modelsa} displays a na\"ive design which assumes that the server
is fully trusted, hence allowing an
unscrupulous operator to return incorrect query results. \Cref{fig:modelsb} displays our system model, in which the server is untrusted.

\begin{figure}[t]
\centering
\subfloat[][\label{fig:modelsa}]{\includegraphics[width=0.4125\linewidth,valign=c]{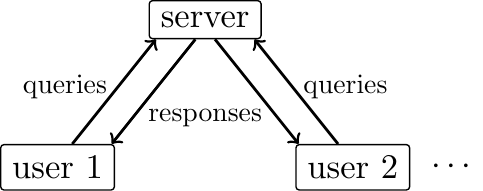}
\vphantom{\includegraphics[width=0.5\linewidth,valign=c]{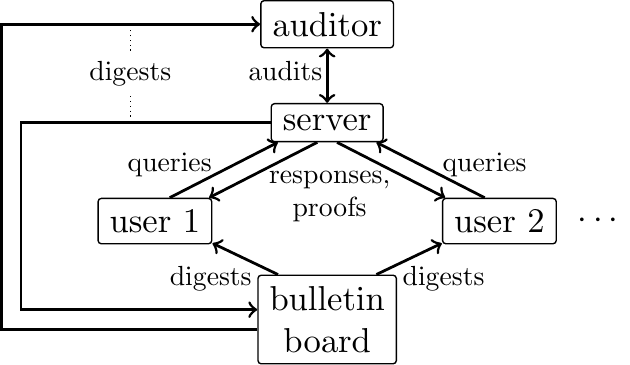}}}
\hspace{0.03\linewidth}
\subfloat[][\label{fig:modelsb}]{\includegraphics[width=0.525\linewidth,valign=c]{motivating_example_3.pdf}}
\caption{Left: system model with a trusted server. Right: \name's system model with an untrusted server.} 

\end{figure}

\subsection{Data Model}
\label{sec:data_model}

We consider a simple relational data model. The schema consists of the following attributes.  

\bparagraph{Time} Time is modeled as a sequence of \textit{epochs}, i.e., time slots such as hours or days, and is represented as
an integer.

\bparagraph{Value} This attribute contains the privacy-sensitive data. For simplicity, we assume that they
are non-negative integers.  %

\bparagraph{ID} This represents the user ID associated with the value. 

\bparagraph{Type} These are string attributes that capture metadata about users.

We denote the number of Type attributes by $m$, and the total number of attributes in our schema by $m'=m+3$.
We assume that there is one data entry per ID per epoch. This assumption prevents a single adversarial user from arbitrarily skewing query results, as we will discuss in \Cref{sec:analysis}. However, a single entity may control multiple IDs, e.g., a single person owning multiple cars for congestion pricing. 

\subsection{Use Case Examples}
\label{sec:use_cases}

In the following, we present three illustrative use cases for a transparent data service. In each case, the users receive \textit{bills} or \textit{rewards} that are dependent on the recorded data, including (potentially) the data of others. In all cases, the Time attribute is the smallest time interval in which the billing or reward rate remains the same. In the following, we discuss the key features of each use case, including the ID, Value and Type attributes, the typical scale of the system, and types of queries. 

\textbf{Smart Grids.} Our first example is a smart grid that enables \textit{peak/off-peak pricing}, i.e., customers are billed at a higher rate for power usage when system-level demand is high. 
In this setting, each \textit{ID} corresponds to a smart meter's serial number, and the \textit{Value} to the customer's power usage during the epoch. Possible \textit{Type} attributes include the customer's geographic location, and the customer's type (industrial, commercial, or residential). The service is maintained by the power retailer. To illustrate the size of a typical smart grid, we consider SP Group, which is the largest power retailer in Singapore with 1.6 million customers in 2022 \cite{spgroup}. For peak/off-peak pricing, we record power usage once per hour. For the Type attributes, we consider the 28 postal code regions of Singapore, and the 3 customer types mentioned above. 

In this setting, the main design goal for a transparent data service is to allow customers to verify their electricity bills. Each customer's bill for a given period $T$ can be expressed as the sum, over each epoch $t$ in $T$, of the electricity price in $t$ multiplied by the customer's power usage in $t$. The first advantage of our data service is that it allows customers to verify their usage and bills. However, a major advantage of \name is that it also enables more advanced pricing methods, such as making the price dependent on the \textit{system-level} power usage \cite{esorics21}. The data service allows the users to verify claims about the system-level through sum queries. Finally, the data service allows for the computation of aggregate statistics, e.g., 1) the average and standard deviation of power consumption of all users within the same region, 2) the maximum and minimum usage in a region during a given period, and 3) the top $5\%$ consumption across all residential users. 

\textbf{Congestion Pricing.} Our second example is a system in which vehicles are charged when they cross designated road sections that are heavily congested during peak periods. To detect vehicles, \textit{gantries} with cameras are placed alongside the designated road sections. 
Each \textit{ID} corresponds to a \textit{pair} of vehicle and gantry IDs, and the \textit{Value} to the number of times the vehicle crossed the gantry during the epoch. Possible \textit{Type} attributes include the gantry ID and the type of vehicle, e.g., car, truck, or motorcycle. 
As an illustrative example, we consider the Electronic Road Pricing (ERP) system in Singapore. As of 2022, the ERP system consists of 77 gantries of which around 20 are located at the entries and exits of the highly congested Central Business District (CBD) \cite{sgcarmart}. 

As in smart grids, the data service allows users to verify their bills even for advanced pricing schemes. For example, the price can be made dependent on the total number of registered vehicles that have crossed a gantry in an epoch, which would act as a proxy for the real congestion in the system. An even more advanced pricing scheme would make the price of entering the CBD dependent on the sum of vehicle entries into the CBD minus the sum of vehicle exits.

\textbf{Digital Advertising.} Our final example is a system that rewards websites who display digital advertisements  \cite{bashir2016tracing}. In particular, a website owner receives a reward whenever an advertisement displayed on the website is clicked.\footnote{This is a simplified version of online advertising, as advertisers and publishers are often interested in more detail than just click-through rates.}
In this setting, \textit{ID} corresponds to a pair of website and ad IDs, and the \textit{Value} to the number of times the advertisement was clicked on via the website (i.e., the click-through rate). Possible Type attributes include the category and size of the website, or characteristics of the advertisement. The advertisement platform operates the server, while the website owners and advertisers monitor their data entries.
The biggest online advertisement platform is Google Ads with millions of registered websites and advertisers. However, there are also digital ad platforms such as sixads \cite{sixads} that serve around 100\,000 websites. 

The data service allows both the advertisers and the websites to monitor the click-through data provided by the advertising platform. This makes it easier to detect misbehavior, e.g., the platform underreporting (or overreporting) the number of clicks to the website owner (or advertiser) for profit. Finally, it enables more advanced reward schemes, e.g., in which the total reward paid by any single advertiser is limited. 

\subsection{Threat Model}
We consider two types of threats.
The first consists of \textit{honest-but-curious} users who follow the protocol but try to learn the
privacy-sensitive data of other users. The second is an \textit{adversarial server} who tries to falsify query results by tampering with the data and/or query
execution. 
This assumption captures unscrupulous server owners, insider threats, and external attacks via software
vulnerabilities. We do not consider threats to privacy that stem from collusion between the server and users, because the server is always free to send raw data to adversarial users out-of-protocol. However, a limited set of adversarial users may collude with the server to falsify query results -- this includes fake users created by the server. 
We discuss how users disseminate information about server misbehavior in \Cref{sec:discussion}.

We assume that all communication between
the server and honest users is done via secure channels. The auditors are only trusted to validate the server's structural properties of the server's ADS, but not individual data entries. 
In fact, the trust assumption for auditors is the same as for users, i.e., any user with large computation and network resources
can also act as an auditor. In practice, we also assume that one ``super'' auditor (e.g., a regulator) has the ability to verify user identities, to prevent service providers from creating an unlimited number of fake users.  Finally, we assume that the public bulletin board -- or, alternatively, the gossip protocol for users -- is trusted and able to detect equivocation. 

\subsection{Requirements}
\label{sec:requirements}
Our goal is to design a system that meets the following requirements. Due to space constraints, we only present
informal definitions of the security requirements, and leave the formal definitions to Section~\ref{sec:analysis} and the appendix.

\bparagraph{(R1) Rich operations for multiple users} The system supports a wide range of operations (or queries) on the aggregate data generated by multiple independent users. 

\bparagraph{(R2) Data privacy} A user can only learn a limited number of other users' values by performing queries.

\bparagraph{(R3a) Data integrity} The server cannot change the data without being detected.  

\bparagraph{(R3b) Transparency} For each supported query, the server cannot convince the user to accept incorrect results computed from incorrect, incomplete, or artificial data.  

\bparagraph{(R4) Efficiency} The computation, storage, and network costs at the server and the user's client are small.
Query overheads grow sublinearly with the number of users and epochs.

\section{Existing Solutions}
\label{sec:survey}
In this section, we discuss existing approaches that partially meet our requirements. We divide them
into three categories: transparency logs, PoLs, and SQL-based authenticated databases. 
Visual representations of various system models, and several ADS designs whose features are incorporated in \name, can be found in \Cref{fig:existing_solutions}. 
We conclude the section by discussing which requirements are met by these approaches.  

\begin{figure*}[t]
\centering
\subfloat[][\mbox{CONIKS}\label{fig:coniks_system}]{\includegraphics[width=0.2625\linewidth]{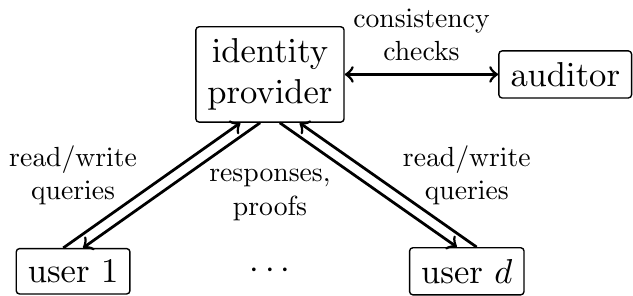}}\hspace{0.05\linewidth}
\subfloat[][IntegriDB\label{fig:integridb_system}]{\includegraphics[width=0.25\linewidth]{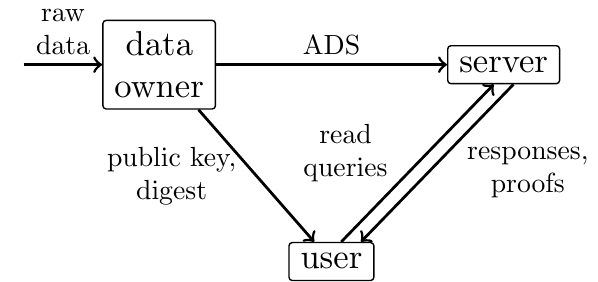}}
\hspace{0.05\linewidth}
\subfloat[][\mbox{FalconDB}\label{fig:falcondb_system}]{\includegraphics[width=0.18\linewidth]{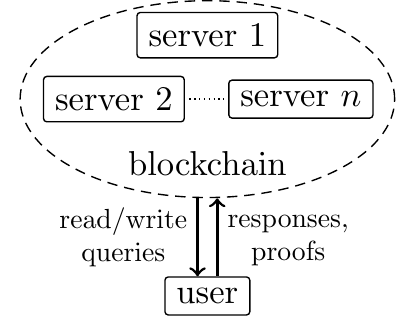}}
\caption{System models of CONIKS, IntegriDB, and FalconDB.}\label{fig:existing_solutions}
\end{figure*} 

\begin{figure*}[t]
\centering
\subfloat[][CONIKS\label{fig:coniks_ads}]{\includegraphics[width=0.375\linewidth]{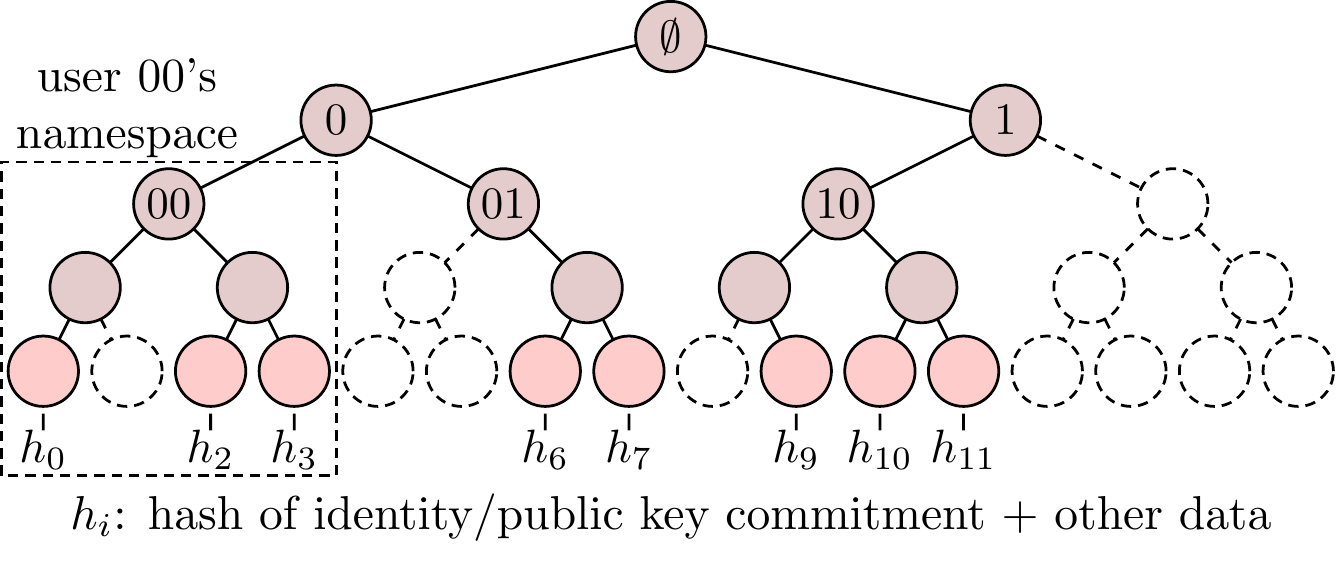}}
\hspace{0.005\linewidth}
\subfloat[][Proofs-of-Liabilities\label{fig:pol_ads}]{\includegraphics[width=0.295\linewidth]{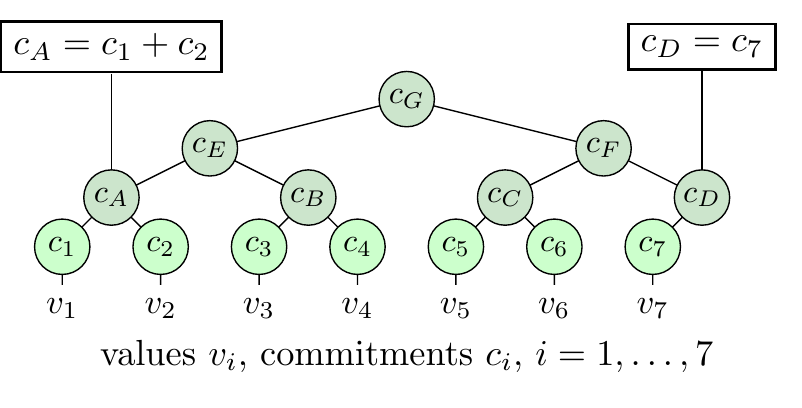}}
\subfloat[][IntegriDB and FalconDB\label{fig:integridb_ads}]{\includegraphics[width=0.295\linewidth]{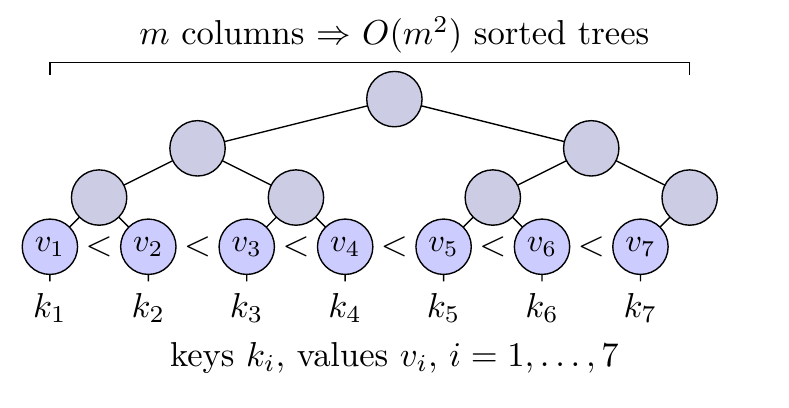}}
\caption{ADS designs of CONIKS, Proofs-of-Liabilities, IntegriDB, and FalconDB.}\label{fig:existing_solutions}
\end{figure*} 

\subsection{Transparency Logs} 
Transparency logs are append-only data structures whose integrity is protected by a Merkle tree. They provide
efficient cryptographic proofs that show, e.g., that an entry is included in the log, or that one log is a prefix
of another log. 

\bparagraph{Certificate Transparency (CT)} CT~\cite{laurie2014certificate} addresses the problem of
compromised certificate authorities by publishing the certificates on a transparency log. 
The system model of CT consists of some certificate authorities
who issue certificates and insert them into the log, and users who search for specific certificates in the log.
CT relies on a monitor to ensure the consistency of the log. The core data structure is a Merkle tree in which the certificates
are hashed and stored at the leaves, and the server signs the root of the tree. CT has been extended to support efficient verification of certificate revocations
\cite{laurie2012revocation,ryan2014enhanced}. It has also been generalized into an abstraction called a
\textit{verifiable log} which is implemented as Google's Trillian \cite{eijdenberg2015verifiable}.

\bparagraph{CONIKS} CONIKS \cite{melara2015coniks} extends CT to support transparent name-to-key bindings. It
allows for efficient proofs of non-inclusion so that users can easily check for unauthorized name-to-key
bindings in their namespaces. CONIKS' system model, depicted in \Cref{fig:coniks_system} for a system with $d$ users, is similar to that of CT, but users are more active in
monitoring their key bindings. CONIKS uses \textit{prefix trees}, depicted in \Cref{fig:coniks_ads}, for efficient
non-inclusion proofs. It hides bindings by storing only their commitments at the leaves. It also hides the
total number of users by adding dummy nodes. 

\bparagraph{Merkle$^2$} Merkle$^2$ \cite{hu2021merkle} extends CONIKS through a data structure that enables
efficient auditing. The data structure combines
a \textit{chronological Merkle tree} with prefix trees. The leaves of the chronological tree only extend to
the right (append-only). Each internal node protects a set of leaves in the chronological tree, and stores the
root of a prefix tree that has the same set of leaves. 

\subsection{Proofs of Liabilities (PoLs)} PoLs \cite{dagher2015provisions} are designed to prove solvency --
i.e., assets being greater than liabilities -- in a setting where users are fully anonymous and their individual assets and liabilities are privacy-sensitive. The main data structure in PoLs is a \textit{Merkle sum tree}, as depicted in \Cref{fig:pol_ads}, that stores additively homomorphic \textit{commitments} to the values of assets and liabilities in the leaves. Intermediate nodes store the sums of the commitments in their children. To show that the sum of assets and liabilities in the leaves is
non-negative, PoLs use \textit{zero-knowledge range proofs}. 
PoLs
were generalized in \cite{ji2021generalized} to use cases beyond proving solvency.

\subsection{SQL-Based Authenticated Databases}
\bparagraph{IntegriDB} The system model of IntegriDB~\cite{zhang2015integridb} is designed for outsourced databases. In particular, the user uploads data
and metadata to an untrusted server, depicted in \Cref{fig:integridb_system}. The server executes user queries and generates proofs based on the metadata
to show that the results are correct. 
IntegriDB supports data insertions, join queries on multiple tables,
multidimensional range queries, and sum, count, average, min, and max queries. IntegriDB creates a sorted tree for each
column pair, resulting in $\frac{1}{2}(m^2-m)$ trees for a table with $m$ columns, as depicted in \Cref{fig:integridb_ads}. In each internal node of the tree, IntegriDB stores a polynomial
over the values in the leaves of the internal node's subtree. The polynomials enable proofs that sum or range queries have been performed
over the correct set of leaves. Meanwhile, the sorted nature of the trees allows users to verify min and max queries. Another SQL-based ADS, vSQL \cite{zhang2017vsql}, supports generic SQL queries and has similar performance as IntegriDB.

\bparagraph{FalconDB} FalconDB \cite{peng2020falcondb} combines IntegriDB with blockchains.
In FalconDB's system model, depicted in \Cref{fig:falcondb_system}, a smart contract maintains the
ADS and ensures that it is globally consistent. Queries are performed directly by the
servers, which run IntegriDB, without going through consensus (except for insertions and removals). Users
verify the results by checking that the ADSs at the servers are the same as in the
blockchain. 

\subsection{Limitations of Existing Solutions}
\label{sec:limitations}
Transparency logs meet our requirements of multiple user support, privacy, integrity, and efficiency. However,
the range of supported operations, namely insertion, deletion, inclusion, and non-inclusion, is too limited to achieve R1. Similarly, PoLs do not achieve R1 as they only support sum queries, but not standard deviations, minima/maxima, or quantiles.

IntegriDB and Falcon achieve R3a, R3b, and R4, but cannot support R1 and R2 simultaneously. 
IntegriDB's system model assumes that a single user generates the ADS correctly before uploading it to the server. As such, IntegriDB cannot be easily extended to support
multiple independent users. In particular, the server can maintain separate databases and ADSs,
but users cannot verify operations on the aggregate data without building the data structures on the entire
data by themselves. The users must therefore know each other's data, i.e., the system cannot simultaxneously achieve R1
and R2. FalconDB supports multiple users, but all data are stored on the blockchain, thus it
does not meet R2. 

\section{\name}
\label{sec:solution}
In this section, we describe \name, a transparent data service that overcomes the limitations discussed in \Cref{sec:limitations}. We start by describing the main building blocks, and then
explain the core ADS and how it supports a rich set of operations. Finally, we
discuss how audits are performed in \name.

\subsection{Preliminaries}
\label{sec:preliminaries}

We use several cryptographic primitives. We only define them briefly here due to space constraints, and
refer readers to the literature for their formal and complete definitions. 

A \textit{hash function} $H$ takes as input a value $x \in \{0,1\}^*$ and outputs a value in
$\{0,1\}^{l_{H}}$, where $l_{H}$ is the output length of the hash function. The function is \textit{collision-resistant} if
the probability of finding two different inputs that produce the same hash output is negligible.

A \textit{commitment scheme} $\textsc{COM}$ consists of two algorithms. $\textsc{COM.Setup}$ takes as input a
security parameter $1^{\kappa}$ and outputs the commitment parameters $P_{c}$. Let $V_c$ and $R_c$ be the sets of all possible data and random values, respectively. $\textsc{COM.Commit}$ takes as
input the parameters $P_c$, a value $v \in V_c$, and a random value $r \in R_c$ (which we also call a \textit{seed} to avoid confusion with data values), and outputs a commitment $c'$. $\textsc{COM}$ is
called \textit{hiding} when $c'$ reveals nothing about $v$, and \textit{binding} when given a commitment of $v$
and $r$, it is computationally infeasible to find another $v'$ and $r'$ that produce the same commitment.
We use the shorthand notation $C(v,r) = \textsc{COM.Commit}(P_c,v,r)$, with $P_c$ set during initialization.
The scheme is \textit{additively homomorphic} if for any $v_0,v_1 \in V_{c}$ and $r_0,r_1 \in R_c$, it holds that 
$$
C(v_0, r_0) + C(v_1, r_1) = C(v_0 + v_1, r_0 + r_1).
$$
A \textit{non-interactive zero-knowledge proof system} consists of three algorithms. $\textsc{NIZK.Setup}$
takes as input a security parameter $1^\kappa$ and outputs system parameters $P_{zk}$. $\textsc{NIZK.Prove}$
takes as input the system parameters $P_{zk}$ and a statement-witness pair $(s,w)$ and outputs a proof $\pi$.
$\textsc{NIZK.Verify}$ takes as input $P_{zk}$, a statement $s$, and a proof $\pi$, and outputs \textsc{true}
or \textsc{false}.  
The proof system $\textsc{NIZK}$ satisfies {\em zero-knowledge} if the generated proofs reveal nothing about
the witnesses, and {\em simulation-extractability} if for any proof generated by the adversary, there exists
an efficient algorithm to extract the corresponding witnesses with a trap door. In \name, we use
$\textsc{NIZK}$ over the following relation for the range proofs: $$ R_{zk}(v_{\max}) = \{(c,v_{\max}), (v,r),
| c = C(v,r) \wedge v \in [0,v_{\max})\}
$$
We can prove statements of the form $v \in [a,b]$ by proving $v - b + K \in [0, K]$ and 
$v - a \in [0,K]$ for large $K$ \cite{camenisch2008efficient}.

A \textit{Merkle tree} is a binary tree in which each node $\node$ stores a hash value $h_{\node}$. The hashes
have the following structure. The leaves contain the hashes of the values stored in the tree. For the internal
nodes it holds that \mbox{$h_{\node} = H(h_{\textsc{left}(\node)} \,|\, h_{\textsc{right}(\node)})$}, where
$|$ represents concatenation, $\textsc{left}(\node)$ and $\textsc{right}(\node)$ respectively return the left and right
child of $\node$ (if there is no child, they return $0$), and $h_0 = (0)^{l_{H}}$. Similarly, a
Merkle \textit{sum tree} contains in each leaf $\node$ a commitment $c_{\node} = C(v,r)$ where $v$ is the leaf's value and $r$ its seed,
and each internal node $\node$ stores \mbox{$c_{\node} = c_{\textsc{left}(\node)} +
c_{\textsc{right}(\node)}$}. Inclusion of a leaf node $\node$ in a Merkle tree can be proven through a \textit{co-path}, i.e., the sibling nodes on the path between $\node$ and the tree's root. The prover can use the co-path and the leaf to rebuild the hash (or commitment) in the root, and compare it to its known value.

A \textit{prefix tree} is a binary tree in which each leaf corresponds to a key-value pair
$(\chi,\phi)$ where $\chi$ is a bit string of length $l_P$. Each internal node $\node$ stores the key $\chi_{\node}$ of
length $l'$, such that $l' < l_{P}$, and its left and right child nodes have the key $\chi_{\node}|0$ and
$\chi_{\node}|1$, respectively. A prefix tree is typically extremely sparse, and we only need to store the internal nodes that are on a
direct path between a leaf and the root (which has key $\emptyset$). 
A prefix tree can be extended to a Merkle prefix tree by including a hash of the last prefix bit and value in each leaf, and the hash of the last bit and child hashes in internal nodes.

\subsection{Data Structure}
\label{sec:ads}
The ADS in \name combines a single Merkle prefix tree (as in \Cref{fig:coniks_ads}) with multiple sorted Merkle sum trees (which combine the key features of Figures~\ref{fig:pol_ads}~and~\ref{fig:integridb_ads}). The \textit{key} of each leaf $\node$ in the
prefix tree corresponds to a unique combination of values of the Time and Type attributes, whereas the
\textit{value} of $\node$ is the root hash of a sorted Merkle sum tree. This Merkle sum tree is constructed
from the Value attributes of the data whose Time and Type attributes are equal to the prefix tree
leaf's key. The tree is sorted by the values of its leaves in ascending order. Each leaf not only stores
the raw value $v$, but also $v^2, v^3,\ldots,v^{z}$ for some system-wide integer $z$. These values enable the computation of advanced statistics (e.g., the standard deviation), while having no impact on the ordering
because $x^k$ is monotonically non-decreasing for $x,k \geq 0$. 

The prefix tree is stored in memory, whereas the full table is stored in a SQL database. We do not keep the
Merkle sum trees in memory except for their root hashes, which are stored in the prefix tree leaves, because 
the Merkle sum trees are easily constructed on-the-fly during queries.

\begin{table}[t]
\centering
\caption{Illustrative data for \Cref{fig:example_tree}.}
\label{tab:example}
\scalebox{0.8}{
\begin{tabular}{|cccc|}
\hline
Time & ID & Type & Value \\ \hline
0 & Alice & residential & 11 \\
0 & Bob & residential & 24 \\
0 & Carol & residential & 13 \\
1 & Alice & residential & 19 \\
1 & Bob & residential & 26 \\
1 & Carol & residential & 27 \\
1 & Dave & residential & 26 \\
1 & Erin & industrial & 36 \\ \hline
\end{tabular}
}
\end{table}

The data structure in \name is initialized as follows. For simplicity, we assume that the server takes as
input a table consisting of initial data from multiple users. Given a table with Time, ID, Type, Value attributes,  the
server generates a random seed $r_i$ for each row $i$ and stores it as a new attribute {\bf Seed}. 
The server then computes the sets of all \textit{unique combinations} of values of the Time and Type attributes; we denote this set by $\mathcal{S}$. For each unique value tuple $s \in \mathcal{S}$, we determine the array $V_s =
(v_{s\,i})_{i=1,\ldots,|V_s|}$ that contains the Value attributes of the rows whose Time and Type attributes match
the elements of $s$. The server sorts $V_s$ in ascending order, i.e., $v_{s\,i} \leq v_{s\,i+1}$ for all $i
=1,\ldots,|V_s|-1$. Let $u_i$ be the ID attribute of the row with corresponding Value $v_i$. For
each $s \in \mathcal{S}$ such that $|V_s| > 0$, the server creates a Merkle sum tree where the $i$th leaf contains the
following values: 
\begin{itemize} 
\item $c_{i\,j} = C(v^j_i,r_i)$ for \mbox{$j \leq z$}, where $v_i \in V_s$ and $r_i$ is the corresponding Seed attribute,
\item $l_{i} = 1 $, and 
\item $h_{i} = H(c_{i\,1}|\ldots|c_{i\,z}|H(u_i,s_0))$.
\end{itemize}
For the empty node $0$, we choose $c_{0} = C(0,0)$, $l_0 = 0$, and $h_{0} = H(0)$. Informally, $c_{i\,j}$ contains the
commitment of the $i$th value (in ascending order) to the power $j$, $l_{i}$ contains the \textit{count},
i.e., the number of leaves in the subtree rooted at the node (which is always $1$ for the leaf nodes),
and $h_{i}$ contains the leaf's hash.
Each internal node $\node$ contains the following values: 
\begin{itemize}
\item $c_{\node\,j} = c_{\textsc{left}(\node)\,j} + c_{\textsc{right}(\node)\,j}$
\item $l_{\node} = l_{\textsc{left}(\node)} + l_{\textsc{right}(\node)}$, and
\item \mbox{$h_{\node} = H(h_{\textsc{left}(\node)}|h_{\textsc{right}(\node)}|c_{\node\,1}|\ldots|c_{\node\,z}|l_{\node})$}.
\end{itemize}
The hash in the root of the sum tree is stored as the value in the prefix tree leaf whose prefix is the bit concatenation of the Time and Type values. After
initialization, the server sends the initial digest $\delta_{0}$, which is the hash of the Merkle prefix tree
root, to the bulletin board. For every value $v_i$ that is inserted, the server sends the random value $r_i$ in the
corresponding commitment to the user, which the user can later use to verify the commitment. 

For example, \Cref{fig:example_tree} depicts the ADS after processing the rows from \Cref{tab:example}, if `residential' and `industrial' are mapped to 0 and 1, respectively. Each sum tree leaf node (denoted by $\textnormal{i},\ldots,\textnormal{viii}$) contains the commitments $c_1\ldots,c_z$ of a Value column entry, and $h$ and $l=1$ as discussed previously. The intermediate sum tree nodes ${A},\ldots,{F}$ contain the values $l$, $c_1\ldots,c_z$, and $h$ computed from their child nodes. The prefix tree nodes $\emptyset,\ldots,\textnormal{11}$ contain hash values, and
each prefix tree leaf corresponds to a unique combination of Time/Type values. In this case, $\mathcal{S} = \{(0, 0),(0, 1),(1, 0),(1, 1)\}$, but there are no rows for which Time and Type respectively equal $0$ and $1$ (`industrial') -- as such, only three leaves (00, 10, and 11) are stored in the prefix tree.

\begin{figure}[t]
\centering
\includegraphics[width=\linewidth]{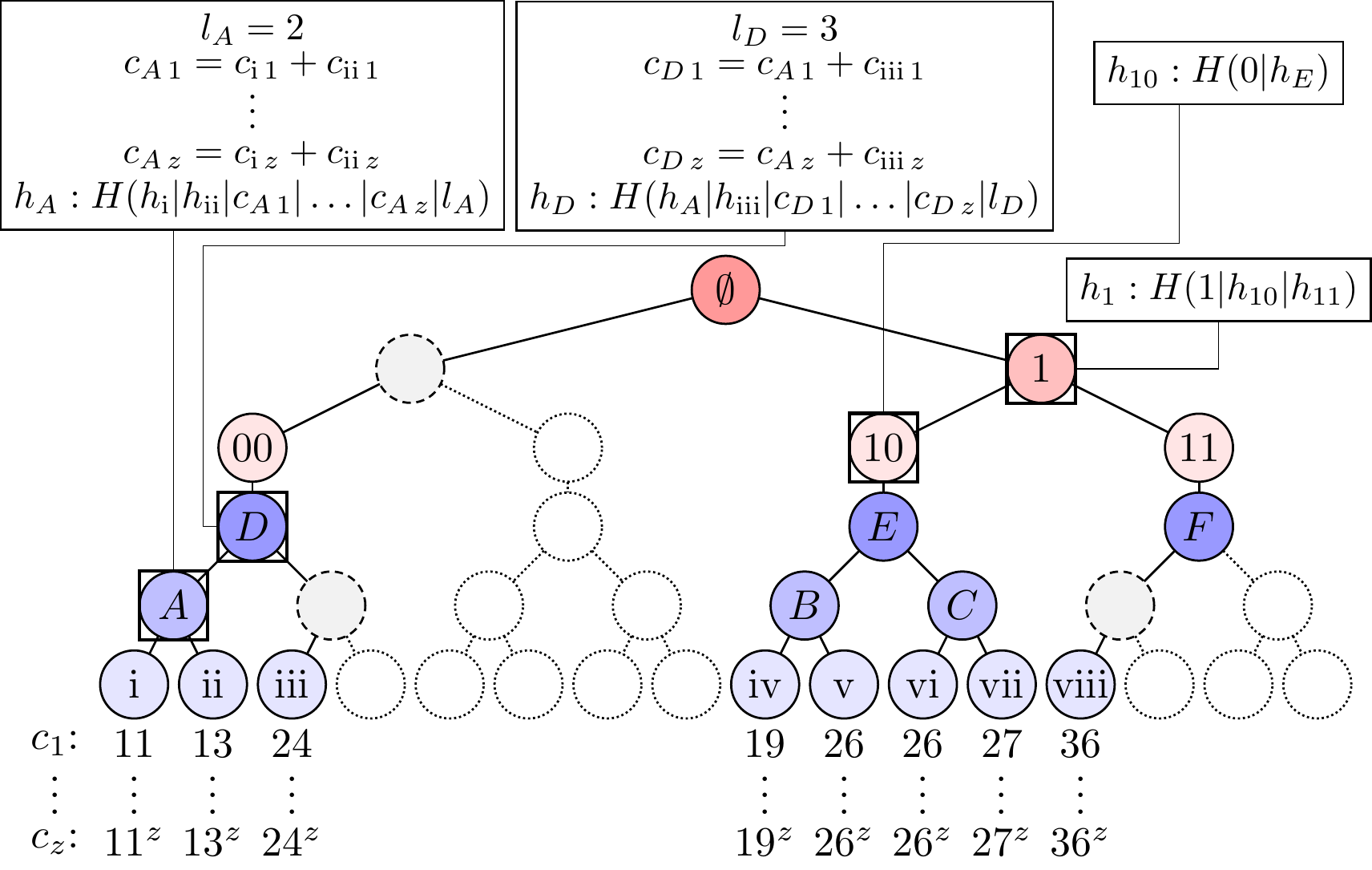}
\caption{\name's ADS after inserting the $8$ rows from \Cref{tab:example}.} 
\label{fig:example_tree} 
\end{figure}

\subsection{Queries}
\label{sec:queries}
\bparagraph{Insert} As mentioned in \Cref{sec:data_model}, only one new entry per user ID can be inserted into the tree per epoch. To insert the entries of epoch
$t>0$, the server computes $\mathcal{S}_t$, the set of unique combinations of values of the 
{Type} attributes, right-concatenated with the epoch $t$. It then determines the value sets $V_s$ for each \mbox{$s \in \mathcal{S}_t$}, and
constructs the Merkle sum trees whose leaves are the sorted values in $V_s$. 
Next, it inserts for each $s$ for which $|V_s|>0$ a new key-value pair $(\prefixkey,\prefixval)$ into the existing prefix tree,
where $\prefixkey$ is the the bit string that represents $s$, and $\prefixval$ is the root hash of the corresponding Merkle sum tree. The server sends
the digest $\delta_t$, i.e., the root of the updated Merkle prefix tree, to the bulletin board. Finally, it sends
to each user $i$ the random value $r_i$.  
Since the leading bits in the prefixes correspond to the Time attribute, the prefix tree is sorted chronologically. As in Merkle$^2$, this allows auditors to efficiency verify that the tree is append-only.

\bparagraph{Look-up} 
The user can search for the value $v$ tied to her ID and seed for a given specific Time attribute $t$. The server executes the query on the SQL database, and if it finds a result then it generates a proof
$\pi = (v, \pi_1, \pi_2)$ as follows. First, it computes a prefix $\chi$ using $t$ and the user's Type attributes.
Next, it constructs the Merkle inclusion proof $\pi_1$ for the leaf $(\chi,\phi)$ in the prefix tree, 
where $\phi$ is the value for key $\chi$. It then produces another inclusion proof $\pi_2$ for the value $C(v,r)$ in the Merkle sum tree whose root is
$\phi$. The user then verifies the proof by requesting the digest $\delta_t$ from the bulletin board,
computing $c=C(v,r)$, and verifying that the inclusion proofs are correct with respect to $c$ and the digest.

If no data entry is found, then the server generates the following non-existence proof: 
\mbox{$
\pi = (\pi',N^*,(h'_{\node})_{\node \in N^*})
$}
First, it computes a prefix $\chi$ as above, and constructs the Merkle inclusion
proof $\pi'$ for the leaf $(\chi,\phi)$ in the prefix tree, where $\phi$ again is the value in the leaf with prefix $\chi$. Let
$N^*$ be the set of leaves in the Merkle sum tree whose root is $\phi$. For each node $\node \in N^*$, the server
also computes $h'_\node = H(u_{i}|t)$. The user's client verifies the proof by checking the
prefix tree inclusion proof, and checking for all leaves $\node \in N^*$ that $h_\node =
H(c_{i\,1}|\ldots|c_{i\,z}|h'_\node)$. Finally, it rebuilds the sum tree root from the leaves and checks that it
matches $\phi$. 

\bparagraph{Range cover} 
Range queries are a subroutine of aggregate queries. Since \name does not reveal individual data
entries to unauthorized users, the server returns the set of prefix nodes that cover all entries in the
specified range. In particular, the query contains $S^* =
(t^{\min}, s^{\min}_1,\ldots,s^{\min}_m,t^{\max}, s^{\max}_1,\ldots,s^{\max}_m)$ where $t^{\min}$ corresponds to the smallest value of the Time attribute in the range, $t^{\max}$ to the largest value, and
$s^{\min}_1, s^{\min}_2, \ldots$ and $s^{\max}_1, s^{\max}_2, \ldots$ correspond to the Type attributes. The server returns a
set containing each prefix node $\node$ for which it holds that the subtree rooted at $\node$ contains the data entries whose Time and Type attributes overlap with $S^*$.  

To produce a proof, the server calls the function $\textnormal{ExtendRangeProof}$ as described in
\Cref{alg:range_proof} in \Cref{sec:pseudocode} with the root of the prefix tree, $\emptyset$, and $S^*$ as input arguments. This
function recursively calls itself on the node's children.
Given a set $N'$, the client requests the digest $\delta_t$ from the bulletin board and verifies the following
properties: $N'$ is properly formed (i.e., all hashes follow from the child hashes), the root
node included in $N'$ has a hash that equals $\delta_t$, and these nodes completely cover $S^*$.  

\bparagraph{Sum/Count/Average/Standard Deviation} A user can query the sum and count of the values over a
given range $S^*$ defined as above. The server executes the query and generates the following proof: 
$$ 
\pi = (N',L',r^*,v^*_1,\ldots,v^*_z,(h'_\node,c_{\node\,1},\ldots,c_{\node\,z},l_\node)_{\node \in L'}).
$$
In particular, it first computes the range cover proof $N'$ by calling ExtendRangeProof from \Cref{alg:range_proof} on the prefix tree root. It then computes the sums $v^*_1,\ldots,v^*_z$ and the total seed $r^*$
of the covered data entries. Let $L'$ be the set of leaves of the prefix tree in $N'$. The server then
determines for each \mbox{$\node \in L'$} the child hash $h'_\node =
h_{\textsc{left}(\node)}|h_{\textsc{right}(\node)}$, the commitments $c_{\node\,1},\ldots,c_{\node\,z}$, and the leaf
count $l_\node$.

Given $\pi$, the user's client first initializes $c^*_{j}=C(0,0)$ for \mbox{$j=1,\ldots,z$},
and $l^*=0$. Next, it verifies the range cover proof $N'$, then checks for every node $\node \in L'$ whether it holds that
 $$ \phi(\node) =
H(h'_\node|c_{\node\,1}|\ldots|c_{\node\,z}|l_\node), $$
i.e., whether $\phi(\node)$, the value stored in prefix tree leaf $\node$, is constructed as expected from the child hash, commitment, and leaf counts.  
If so, it updates $c^*_{j}$ to $c^*_{j} + c_{n\,j}$ for all $j=1,\ldots,z$ and $l^*$ to \mbox{$l^* + l_\node$}. Finally, it checks that $c^*_j = C(v^*_j,r^*)$ for all $j=1,\ldots,z$. If so, the user can compute statistics such as the sum $v^*_1$, count $l^*$, and average $v^*_1/l^*$. It also enables the computation of more complex query results such as the sample standard
deviation, i.e., $\sqrt{(v^*_2 - (v^*_1)^2/l^*)/(l^*-1)}$.

\Cref{fig:example_sum} visualizes the response to a sum query over all values in \Cref{fig:example_tree}. The server first reveals the range proof, which is the entire prefix tree as the sum query covers the entire dataset. Next, the server reveals $l^* = 8$, $v^*_1 = 182$, and $v^*_2=4604$, which are verified by the client. This allows the client to compute, e.g., the average (22.75) and the sample standard deviation ($\approx$8.137).

\begin{figure}
\centering
\includegraphics[width=\linewidth]{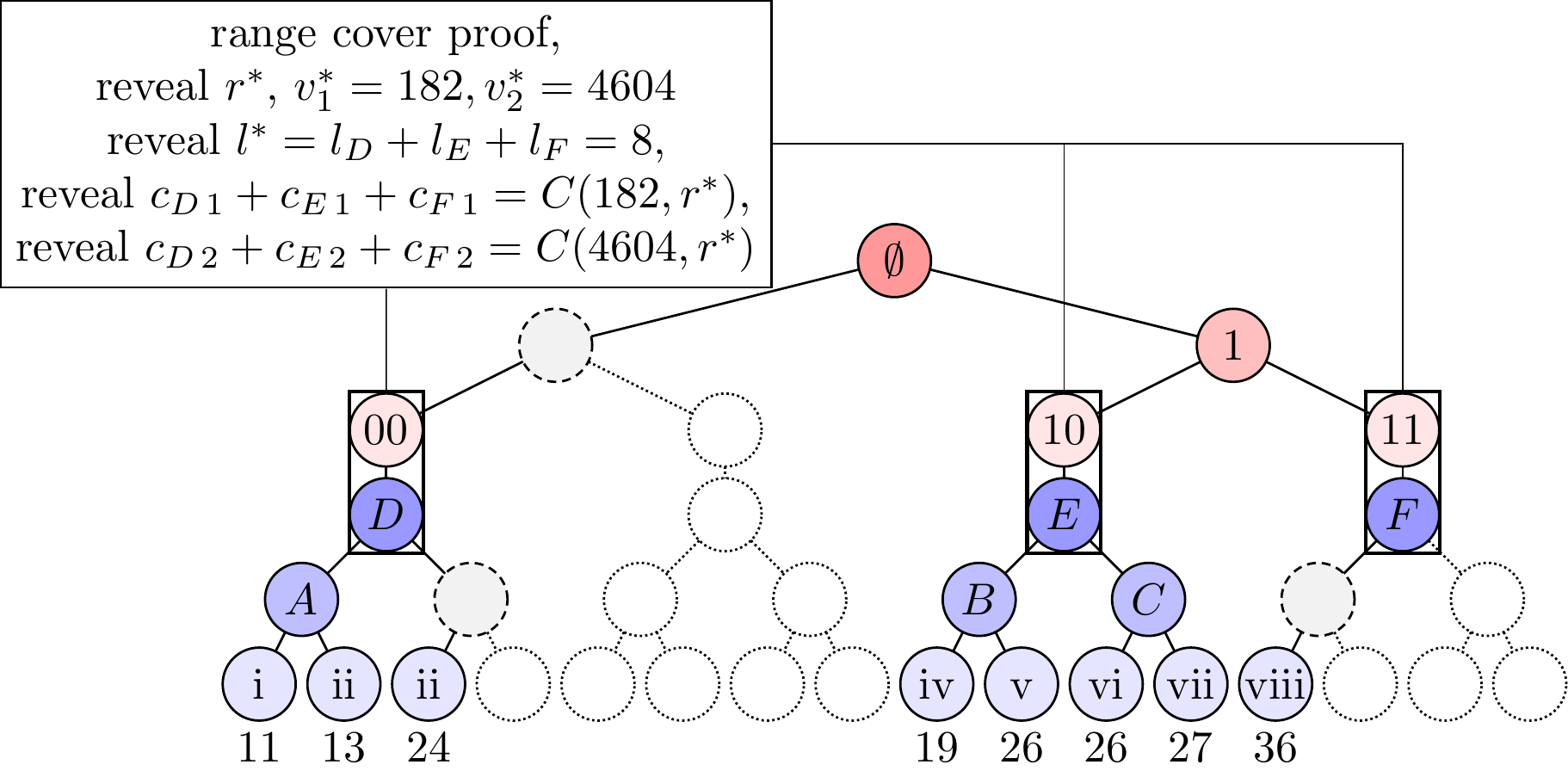}
\caption{Example of a sum query in the ADS of \Cref{fig:example_tree}.} 
\label{fig:example_sum} 
\end{figure}

\bparagraph{Min/Max}
A user can request the min value over a range $S^*$ at epoch $t$, after which the server returns a proof 
$$ \pi =
(N',L',v^*,\node^*,(c_{i\,1},\ldots,c_{i\,z},h'_{\node}, \pi'_\node, \pi^*_\node)_{\node\in L'}). $$ 
Here, $N'$ is the proof of the range query, $L'$ the set of prefix leaf nodes in $N'$, $v^*$
the minimum value in the range, and $i^*$ the index of a prefix
tree leaf whose sum tree contains $v^*$. For each node $i \in L'$, $c_{i\,1},\ldots,c_{i\,z}$ are the  commitments of the value in the \textit{leftmost} leaf in the sum tree of prefix tree leaf $\node$ -- since the
nodes are sorted, the leftmost leaf \textit{must} have the smallest value in the sum tree. The hash $h'_i$ equals the hash $H(u|t)$, such that $u$ and $t$ are the user ID and time in the leftmost leaf. Each inclusion
proof $\pi'_\node$, $\node \in L'$, asserts that the sum tree leaf with hash $H(c_{i\,1}|\ldots|c_{i\,z}|h'_{\node})$ is included in the sum tree whose root is stored in prefix tree leaf $i$. The zero-knowledge range proofs $(\pi^*_\node)_{\node \in L'}$ are as follows: if $\node = \node^*$, then the underlying value in the leftmost leaf must equal $v^*$, so $\pi^*_\node$ asserts that $c_{i\,1}$'s
underlying value is in the range $[v^*,v^*+1)$.\footnote{Alternatively, the server can reveal the random value associated with $v^*$.} Otherwise, $\pi^*_\node$ asserts that this value is in
the range $[v^*,\infty)$.

Given $\pi$, the user's client first checks whether the range cover proof $N'$ is correct. It then
 verifies for all $\node \in L'$ that $\pi_{\node}'$ is a valid inclusion proof for the leaf with hash $H(c_{i\,1}|\ldots|c_{i\,z}|h'_{\node})$, and that this leaf is indeed the leftmost leaf in the path. Next, it verifies the range proofs. By
treating $\node^*$ as a separate case, we guarantee that at least one sum tree contains the value $v^*$ in the
leftmost leaf. If all checks succeed, the user accepts $v^*$ as the min value.
The above is straightforwardly generalized to max queries by proving that the value in the \textit{rightmost} sum tree leaves is
at most equal to the maximum value $v^*$.

\bparagraph{Quantile} A user can request a quantile by specifying a range $S^*$ and a parameter
$q \in [0,1]$ that indicates which quantile to compute. A value $x$ is a $q$-quantile of a set $V$ with $n$ entries
if there are at least $nq$ entries in $V$ whose value is at most $x$ and at least $n(1-q)$ entries whose value is
at least $x$. 
The server executes the query and returns the proof
$$
\pi=(N',L',v^*,(\pi'_{\textsc{left}\,\node}, \pi^*_{\textsc{left}\,\node},\pi'_{\textsc{right}\,\node}, \pi^*_{\textsc{right}\,\node})_{\node \in L'}).
$$
Here, $N'$ is the proof of the range query, $L'$ the set of leaf nodes in $N'$, and $v^*$ the result of the query. For each $\node \in L'$, the
server determines $\textsc{leftmost}(\node,v^*)$, which is the leftmost leaf in the sum tree corresponding to $\node$
whose value is at least $v^*$, if such a leaf exists. It computes the inclusion proof
$\pi'_{\textsc{left}\,\node}$ for this leaf, and a zero-knowledge proof $\pi^*_{\textsc{left}\,\node}$ asserting that its underlying
value is at least $v^*$. Next, it determines $\textsc{rightmost}(\node,v^*)$, which is the rightmost leaf
in $\node$'s sum tree whose value is at most $v^*$, if such a leaf exists. It then
computes the inclusion proof 
$\pi'_{\textsc{left}\,\node}$ for this leaf, and a zero-knowledge proof $\pi^*_{\textsc{left}\,\node}$ asserting that
its underlying value is at most $v^*$.

Given $\pi$, the user first checks that the range cover proof $N'$ is correct. It then
iterates over the nodes $\node \in L'$, where $L'$ is the set of leaf nodes in $N'$, and verifies their 
inclusion and range proofs. The co-path of the Merkle tree inclusion proof includes the leaf count
$l_{{\node}'}$ for each node ${\node}'$ on the co-path; therefore, the user knows for each node on the co-path
whether its leaves are to the right of $\textsc{leftmost}(\node,v^*)$. Let $L^{\textsc{left}}(\node)$ be the number of leaves to the
right of node $\textsc{leftmost}(\node,v^*)$ in the sum tree. Let  
$L^{\textsc{right}}(\node)$ be the number of leaves to the
left of $\textsc{rightmost}(\node,v^*)$ in the sum tree. The user verifies that $\sum_{\node\in N'}
L^{\textsc{left}}(\node) \geq nq $ and \mbox{$\sum_{\node\in N'} L^{\textsc{right}}(\node) \geq n(1-q)$}, and accepts the
result $v^*$ if the verification is successful. 

\Cref{fig:example_quantile} visualizes the response to a query for the median (i.e., the $\frac{1}{2}$-quantile) over all values in \Cref{fig:example_tree}. In this case, the server can choose any value in the interval $[24,26]$ as a valid median, e.g., $v^*=26$. The server first shows that the rightmost leaf in the subtree for prefix `00', and the third leaf for prefix `10', have $v\leq$26. The client uses the $l$-values in the inclusion proofs to determine that there are at least 4 nodes to the left of these nodes, which means that at least 6 leaves have $v\leq$26. Finally, the server shows that the second leaf for prefix `10' and the leaf for prefix `11' have $v\geq$26. As there are 2 leaves to the right of the former leaf, there are at least $4$ leaves with $v\geq$26, which proves that $v^*=26$ is a valid median.

\begin{figure}
\centering
\includegraphics[width=0.85\linewidth]{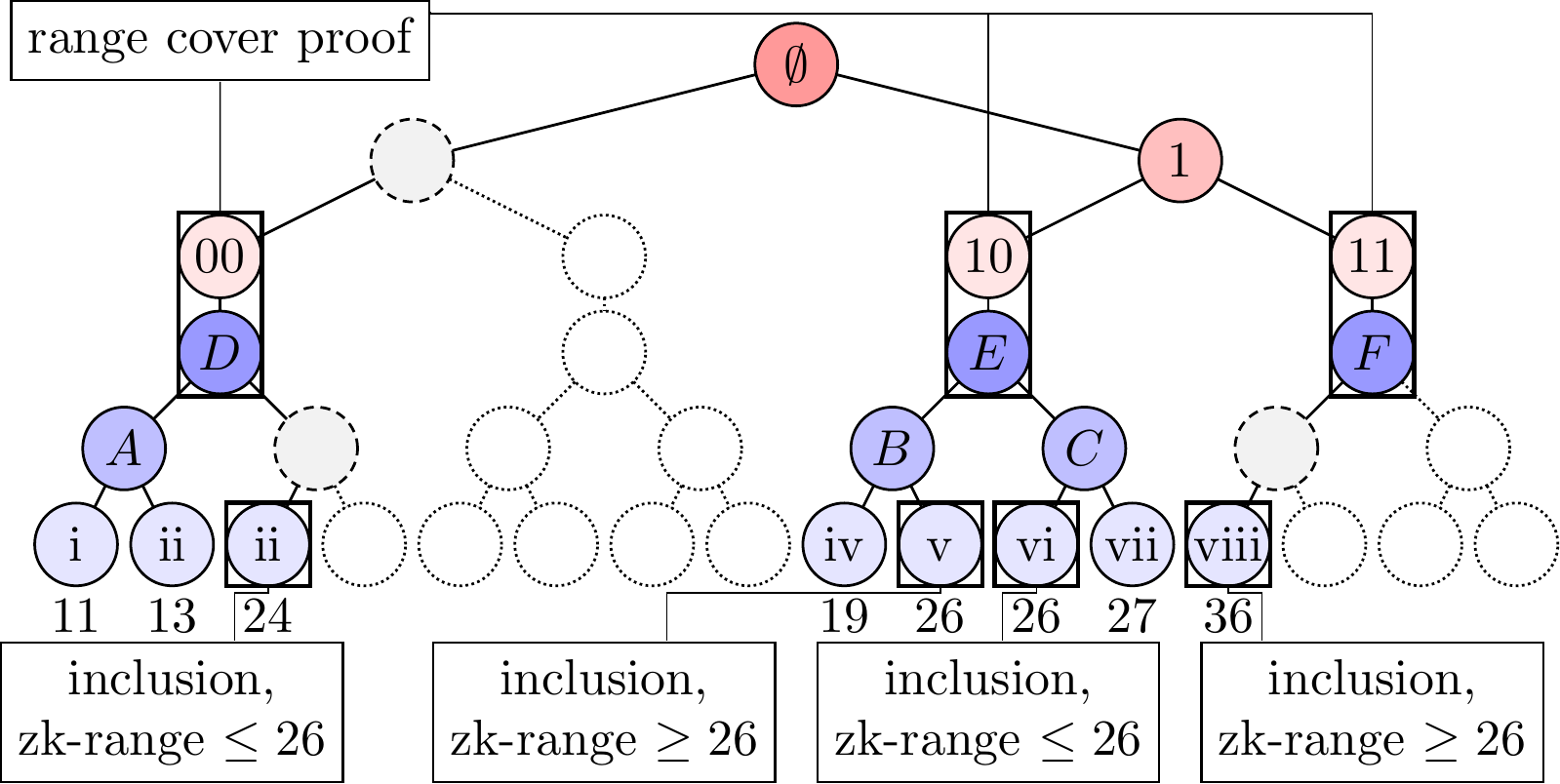}
\caption{Example of a quantile query in the ADS of \Cref{fig:example_tree}.} 
\label{fig:example_quantile} 
\end{figure}

\subsection{Auditing}
The auditor requests the server to generate proofs that show that the tree is append-only and that the sum
trees are sorted by the leaf values. 
For the proof that the tree at epoch $t$ is constructed from the tree at epoch $t'<t$ in an append-only manner,
the server includes each prefix tree leaf $\node$ whose Time attribute $s'_{\node}$ has $t' < s'_{\node} \leq t$. Furthermore, it sends a set of internal nodes from the tree at $t'$ such that the root of the tree at $t$ can be rebuilt from the hashes in these nodes. (This is analogous to the procedure in Section~V.A of \cite{hu2021merkle}.) The auditor then
requests the digests $\delta_{t'}$ and $\delta_t$ from the bulletin board and checks that it can rebuild the
two trees using the new leaves and the old internal nodes. To prove that the leaves in a sum tree with values $v_1,\ldots,v_{n^*}$ are sorted, the server does the following. For each
$i=1,\ldots,n^*-1$, it generates a zero-knowledge range proof that $v_{i+1} - v_{i} \geq 0$ (this is possible because \textsc{NIZK}'s underlying commitment scheme is additively homomorphic), and sends the proofs to the auditor.

\section{Analysis}
\label{sec:analysis}
In this section, we discuss how \name meets the requirements presented in \Cref{sec:requirements}. The server
handles data from multiple users and supports a broad range of queries -- i.e., all single-table queries supported by the
baseline approaches, and additionally quantile queries. In order words, \name satisfies requirement R1. In the
following, we discuss \name against requirements R2, R3a, R3b, and R4 in more detail. 
Due to space limitations, the formal security definitions can be found in the appendix, which we mention in the text when needed.

\subsection{Security (R2, R3a, and R3b)}
\bparagraph{Privacy} We discuss the impact on privacy of each of the supported operations separately.

\textit{Look-up: } an inclusion proof reveals a single value in a sum tree, but only if the user already knew that this value was included. A non-inclusion proof reveals a set of commitments.

\textit{Audit: } reveals the prefix tree and the commitments in the sum trees, neither of which are
privacy-sensitive.

\textit{Sum/count/average:} reveals the sums over the values in the sum trees that correspond to the prefix leaf nodes in the requested range, but no individual values.

\textit{Min/max:} reveals one unique value, namely the minimum/maximum value across the requested range, and the sum tree that contains this value.

\textit{Quantiles:} reveals a value that may correspond to a single unique value in a tree, depending on the values in the requested range and the choice of the server (e.g., if the requested range contains the values $(1,2,3,4)$, then any value in the interval $[2,3]$ is a median, but if it contains $(1,2,2,3)$ then $2$ is the unique median which is therefore revealed). However, it does not reveal which tree(s) (if any) contain this value.

\name queries can leak other values in two cases: the sum if the number of values in a sum tree is very low, and the quantile query if a user is allowed to perform an unlimited number of queries. This follows from the following two theorems, which are proved in \Cref{sec:security_model}.

\begin{theorem}
If a sum tree has $n$ leaves and $f$ values are known by a coalition of users, then the remaining $n-f$ values can be decrypted if and only if $n-f < 2$, regardless of many queries are executed.  
\label{thm:limited_sum}
\end{theorem}

\begin{theorem}
If a sum tree has $n$ leaves and at most $f$ arbitrary quantile queries can be performed, then at most $f$ values can be decrypted. 
\label{thm:limited_quantile} 
\end{theorem}
The privacy properties are proven in \Cref{sec:security_model} through a formally defined \textit{game}, $\GameP^{\Privacy}_{\TAP,\Adversary}$, in which the adversary $\Adversary$ wins if she is able to distinguish the user behind the revealed values. Given Theorem~\ref{thm:limited_sum}~and~\ref{thm:limited_quantile}, privacy can be achieved in practice if the server refuses to respond to sum queries over trees that have a limited number of leaves, and by restricting the range of values $q$ over which $q$-quantiles can be performed. Since the server can efficiently prove claims about the number of leaves in the sum trees (which is the same as proving the result of a count query), the server cannot refuse to respond to valid queries without being detected. In this case, the sum and quantile queries reveal at best a limited number of values. The other query types either reveal one (min/max) or zero (count) values per query, so in this case a limited number of values are returned by default. Note that  the querying language in \name does not allow users to query sums of subsets within a single sum tree, so a sum query reveals nothing about a sum tree except the sum over \textit{all} of its values. Finally, a prohibitively large number of look-up queries would be necessary to reveal a value that is unknown to the user, as the user would have to guess both the value and the random seed correctly. 

In \Cref{sec:differential_privacy}, we discuss strategies to achieve stronger privacy at the expense of transparency by adding noise to query results. If this noise has a protocol-wide bound $\bound$, then the server can efficiently prove that the claimed result is no further than $\bound$ away from the true value using zero-knowledge range proofs for all types of queries. If the noise is unbounded, then we can prove $\epsilon$-differential privacy \cite{dwork2006calibrating,shi2011privacy}. However, unbounded noise allows a malicious server to arbitrarily distort results, which violates transparency (the threat model for differential privacy assumes that the server is honest). On the contrary, if the noise is bounded then the mechanism may satisfy the notion of $(\epsilon,\delta)$-differential privacy. A summary can be found in \Cref{tab:noise_overview}.

 \begin{table}[t]
\centering
\caption{Overview of the impact of adding noise to measurements on integrity and privacy.}
\label{tab:noise_overview}
\scalebox{0.8}{
\begin{tabular}{C{0.275\linewidth}C{0.375\linewidth}C{0.4\linewidth}}
noise type & transparency & privacy \\ \toprule
no noise / return true result & query result is always provably accurate & limited number of values are revealed \\ \midrule
add bounded random noise & query result is accurate within some margin & $(\epsilon,\delta)$-differential privacy, $\delta>0$ \\ \midrule
add unbounded random noise & results can be arbitrarily distorted & provable $\epsilon$-differential privacy \\ \bottomrule
\end{tabular}
}
\end{table}

\bparagraph{Integrity} Honest users monitor the inclusion of their data in each epoch by issuing look-up queries for their own entries and verifying the inclusion proofs. Therefore, if the server tries to add incorrect
data for an honest user in an epoch, then it will be detected. Furthermore, the server cannot modify data from previous epochs, as this
will cause the extension proof to fail during an audit. In \Cref{sec:security_model}, this is formalized using the game $\GameP^{\Transparancy}_{\TAP,\Adversary}$, in which the adversary wins if she able to convince the user to accept a modified value, and which is used to prove both integrity and transparency.

\bparagraph{Transparency} For each type of query supported by \name, the user can verify that the result is
correct. For look-up queries, the user is provided with inclusion and non-inclusion proofs. For sum queries, the user is given inclusion and completeness proofs of the relevant nodes in the prefix tree, with which it
can verify the sum by exploiting the additively homomorphic  property of commitments in \name. For min, max, and quantile proofs, the
user is given zero-knowledge range proofs that assert that the committed value is greater than all,
none, or a specified number of the total nodes in the specified range. As the cryptographic primitives are
secure, the server cannot convince users to accept incorrect results. 

For adversarial or fake users who collude with the server, the server is free to add arbitrary data. As such, sum, average, min, and max queries can be arbitrarily distorted: for example, if the true sum is $v$ and the server wants to increase this to \mbox{$v'>v$}, then the server can select a single adversarial user $i$ whose true value is $v_i$ and increase it to $v' - v + v_i$. However, to distort a $q$-quantile query over a range with $n$ values, a server would require collusion with at least $\max(q,1-q)n$ users. As such, quantile queries are inherently more robust than sum, min, and max queries. To mitigate the impact of adversarial users on sum queries, the server can make the decision to bound all individual values by an individual bound $\gamma>0$. This may be of interest in use cases where extremely high values are unlikely -- e.g., the smart grid or congestion pricing use cases of \Cref{sec:use_cases}. In such, cases, the degree to which sum query results can be distorted is also limited: if $f$ out of $n$ users are adversarial, then the maximum possible distortion is $f\gamma$. Auditors can check that all individual measurements are below $\gamma$ through zero-knowledge range proofs. A similar approach was recently proposed in \cite{esorics21}.

\subsection{Performance (R4)}
\label{sec:performance_analysis}

Table~\ref{tab:comparison_asymptotic} compares the asymptotic costs of queries in \name to related systems. Here, $n$ denotes the number of data rows, $m$ the number of Type attributes (i.e., columns), $d$ the number of users, and $t$ the number of
epochs (i.e., \mbox{$n = O(td)$}). For queries that compute an aggregate over a range, $w$ denotes the number of epochs in the
range. %

\bparagraph{Storage} For storage costs, we consider both the size of the ADS and the underlying dataset, which has size $O(mn)$. \name stores only the prefix tree, which has $O(n)$ nodes, in memory and generates sum trees on-the-fly. The storage costs for \name are
smaller than those in the other systems except for CT. CONIKS regenerates the entire tree after each epoch,
so the storage cost of the ADS is $O(tn)$ \cite{hu2021merkle}. Merkle$^2$ uses a bigger tree than \name
because the internal nodes of Merkle$^2$'s chronological tree store prefix trees. IntegriDB and FalconDB create a
sorted tree for every combination of columns, leading to a total storage cost of $O(m^2n)$.  

\bparagraph{Insertion} Insertion is performed once per epoch with up to $d$ entries. The server needs to insert $O(d)$ new leaves into the prefix tree and a commitment tree has to be built for each prefix
tree leaf. Since the cost of inserting a new entry into the prefix tree is $O(\log n)$, the total cost is $O(d
\log n)$. \name has a higher insert cost than CT, CONIKS, and for reasonable values of $n$ also than
Merkle$^2$. However, \name performs better than IntegriDB and FalconDB because they need to perform one insert for
each of the $m^2$ trees in their ADS.

\bparagraph{Inclusion proof} To prove that a data entry exists in the tree, the server has to generate an
inclusion proof for a prefix tree leaf, generate the corresponding sum tree, and generate an inclusion proof
of the data value in the sum tree. The cost of the first and third steps is $\log n$, and the worst-case cost
of the second step is $O(d)$. Thus, the total cost is $O(d+\log n)$. CT, CONIKS, Merkle$^2$ are
 faster at generating the proof because they do not have to rebuild subtrees. In other words, \name can
improve at this cost by keeping all the sum trees in memory, at the cost of storage. However, we prioritize
storage reduction by generating the sum trees on-the-fly, as look-ups are normally only performed once per user per epoch.

\bparagraph{Non-inclusion proof} Non-inclusion proofs are similar to inclusion proofs, except that all leaves in the sum tree are included, of which there are $d$ in the worst case. The asymptotic cost is therefore \mbox{$O(d+\log n)$}. \name performs better than
CT, IntegriDB, and FalconDB for this proof. However, it performs worse than CONIKS and Merkle$^2$. This is to be
expected as non-inclusion proofs are heavily optimized in these two systems because it is one of their main use
cases. 

\bparagraph{Auditing} To audit a single epoch, the auditor has to verify the append-only proof and range
proofs for the sum trees. The former costs $O(d \log(n))$, and the latter costs $O(d)$. Therefore, the total
cost is $O(d\log(n))$. Only CT and Merkle$^2$ have built-in support for audits. In CONIKS, the tree is
rebuilt in each epoch, which makes auditing difficult. IntegriDB and FalconDB also do not support audits
beyond checking all entries in the $m^2$ trees. Asymptotically, \name is as efficient as CT and Merkle$^2$ --
however, it requires verifying $O(d)$ zero-knowledge range proofs, which are considerably more expensive to verify than Merkle tree
inclusion proofs. 

\bparagraph{Sum} To generate the proof for the sum query, the server first computes the range cover proof $N'$. The
worst-case number of prefix tree leaves in $N'$ is $w d$, as there are at most $d$ leaves per
epoch, and there are $w$ epochs. For each leaf node in the range, its $\log(n)$ parents are included in $N'$,
Therefore, the cost of generating the proof is $O(w d \log n)$. 

In IntegriDB and FalconDB, operations are performed on $m+1$ trees, one for each of the
Type attributes and one for the Time attribute. In each tree, the sum is calculated from the polynomials stored in the internal nodes that form the minimal covering set of the leaves in the range. There are $O(\log(wd))$ of such internal nodes in total. The total cost is therefore at least $O(m \log(wd))$, although the results in \Cref{sec:evaluation} suggest that the processing costs at the server also depend on $n$, which is why the entries for the sum in IntegriDB and FalconDB are marked with an asterisk in \Cref{tab:comparison_asymptotic}.

\bparagraph{Min/Max} The cost of min or max queries is similar to that of sum queries since they are based on
the same range cover proof, except that for each sum tree the server also needs to generate a range proof and an
inclusion proof. The asymptotic cost is still $O(wd \log n)$. For IntegriDB and FalconDB, the cost
is the same as for the sum, i.e., $O(m \log(wd))$.

\bparagraph{Quantile} This query requires at most two range and inclusion proofs per sum tree and a range cover proof -- its asymptotic cost is therefore also $O(w d \log n)$,
and no other baselines support this type of query.

\begin{table*}[htp]
\caption{Asymptotic costs of \name versus other systems -- here,  $n$ is the number of data entries, $m$ the number of columns, $d$ the number of
users, $t$ the number of epochs, and $w$ the number of epochs in the queried range. The asterisks (*) indicate that although the number of tree nodes returned for the proofs of sum, min, and max queries in IntegriDB is sub-linear in $d$, this is not necessarily the case for the size of the data stored in those nodes.}
\begin{center}
\scalebox{0.8}{
\begin{tabular}{c|c|ccccccc}
 \multicolumn{1}{c}{} & \multicolumn{1}{c}{} & \multicolumn{7}{|c}{asymptotic processing costs} \\ 
 & & insert $d$ rows& inclusion & non-inclusion & auditing & & \\ 
 & storage & in epoch & 1 row & 1 row & 1 epoch & SUM & MIN/MAX & quantiles \\ \toprule
CT 
		& $O(m n)$
		& $O(d \log n)$
		& $O(\log n)$ 
		& $O(n)$ 
		& $O(d \log n)$ 
		& ---
		& ---
		& ---
		\\ \midrule
CONIKS 
		& $O(m n + t n)$ 
		& $O(\log n)$
		& $O(\log n)$ 
		& $O(\log n)$
		& ---
		& ---
		& ---
		& ---
		\\ \midrule
Merkle$^2$ 
		& $O(nm + n \log n)$ 
		& $O(d \log^2 n)$
		& $O(\log^2 n)$ 
		& $O(\log^2 n)$ 
		& $O(\log n)$ 
		& ---
		& ---
		& ---
		\\ \midrule
IntegriDB 
		& $O(m^2 n)$
		& $O(d m^2 \log n)$
		& $O(m^2 \log n)$ 
		& $O(m^2 n)$ 
		& --- 
		& $O(m\log (wd))^*$
		& $O(m\log (wd))^*$
		& ---
		\\ \midrule
FalconDB 
		& $O(m^2 n)$  
		& $O(d m^2 \log n)$
		& $O(m^2 \log n)$  
		& $O(m^2 n)$ 
		& --- 
		& $O(m\log (wd))^*$
		& $O(m\log (wd))^*$
		& ---
		\\ \midrule
\textbf{\name}
		& $O(m n)$
		& $O(d \log n)$
		& $O(d + \log n)$
		& $O(d + \log n)$
		& $O(d + \log n)$ 
		& $O(w d \log n)$
		& $O(w d \log n)$
		& $O(w d \log n)$
		\\ \bottomrule
\end{tabular}
}
\end{center}
\label{tab:comparison_asymptotic}
\end{table*}

\begin{table}
\centering
\caption{Time and storage cost of cryptographic operations.} 
\scalebox{0.8}{
\begin{tabular}{c|cc}
operation & time cost (ms) & storage cost (B) \\ \toprule
\textsc{NIZK.Setup} & 11.450 & 22190 \\
\textsc{NIZK.Prove} & 134.780 & 74413 \\
\textsc{NIZK.Verify} & 70.040 & --- \\
generate commitment & 0.144 & 162 \\
sum two commitments & 0.010 & 165 \\ \bottomrule
\end{tabular}
}
\label{tab:crypto}
\end{table}

\section{Evaluation}
\label{sec:evaluation}
In this section, we evaluate the practical performance of \name. We first describe our implementation of \name, then
discuss our experimental setup, and finally present the empirical results. We conduct four types of
experiments: microbenchmarks on a single machine, end-to-end experiments with the user and server on different
machines, a performance comparison against two related baselines, and a scalability experiment that explores the limits of \name. The final three sets of experiments were run on Amazon Web Services (AWS) EC2. 

\subsection{Implementation \& Set-Up}

We have fully implemented \name in Go, and made the source code available at \gitrepo.  We base  
our prefix tree implementation on Merkle$^2$. For commitments and zero-knowledge range proofs, we use the
\texttt{zkrp} library from Morais et al.\ \cite{morais2019survey,zkrpgithub}. This
library uses Bulletproofs \cite{bunz2018bulletproofs} for zero-knowledge range proofs and Pedersen commitment
with the \texttt{secp256k1} elliptic curve \cite{brown2010standards} for additively homomorphic commitments. We use Go's MySQL module for the database backend.
We use the latest versions of the reference implementations of IntegriDB and Merkle$^2$ as of
June 2022 \cite{integridbgithub, merklesquaregithub}.
We use Merkle$^2$ as a baseline for look-up queries and audits, and IntegriDB -- despite having a different system model and support for a broader range of SQL queries -- for aggregate queries, as it is the most efficient approach with publicly available code that we are aware of. Although vSQL \cite{zhang2017vsql} reported comparable performance on a more general class of SQL queries than IntegriDB, we have not included vSQL as a baseline because its implementation is not publicly available.

The microbenchmarks were run on a MacBook Pro laptop with a 2.4 GHz Quad-Core Intel Core i5 processor and 16 GB
of RAM, with iOS 11.6. For the end-to-end experiments, we ran the server on a {\em t2.xlarge} instance and the
user on a {\em t2.micro} instance. For the comparison against IntegriDB and Merkle$^2$, we run all three
systems on {\em t2.xlarge} instances. The scalability experiments were run on a 16-core {\em m5.4xlarge} instance. For some scalability experiments, we took the average over multiple runs to reduce the impact of random noise. To aid reproducibility, we have made an AWS virtual machine image that is set up for the scalability experiments publicly available with identifier 
\href{https://aws.amazon.com/console/}{\texttt{ami-0935fdbddc542254e}}. 
The costs of various cryptographic operations in our system, when executed on the laptop, are shown in \Cref{tab:crypto}. 

\subsection{Microbenchmarks}
We first evaluate the bandwidth costs of the different queries on a single machine. For this experiment, we
consider $d=100$ users who each insert a new value per epoch. We have two Type columns: \textit{region} and
\textit{is\_industrial}. Each user is randomly assigned to one of 10 regions, and 20\% of the users have \textit{is\_industrial} set to 1 and the others have it set to 0. 

Figure~\ref{fig:different_queries_breakdown} compares the proof sizes of different queries, including 
audits. We observe several groups. The look-ups have the smallest proofs, as they only consist of two
Merkle tree co-paths. They are followed by the sum, count, and average, whose proofs include all of the prefix
tree leaves in the range query.  Next are the min and max, which add range proofs for each sum tree.
For quantile queries, we observe some difference between the median and the 5th percentile -- the former has a
larger proof size. The reason is that for each subtree, the query requires only a single inclusion and range proof
if all of the sum tree's values are either greater or smaller than the quantile, and two inclusion and range proofs otherwise.
  The former scenario is more likely for the 5th percentile. The proof for audits is the largest, as it contains $O(d)$
  range proofs. Even in this case, the size remains modest at $10$MB per epoch. 

\begin{figure}[t]
\centering
\subfloat[][]{\includegraphics[width=0.85\linewidth]{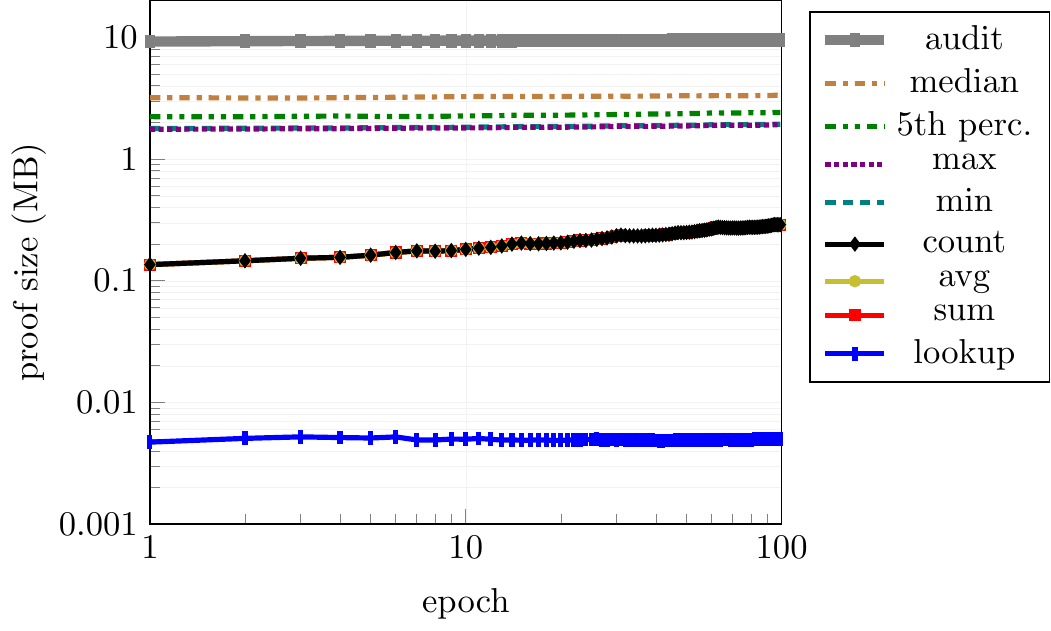}
\label{fig:different_queries_breakdown}
}
\\
\subfloat[][]{\includegraphics[width=0.9\linewidth]{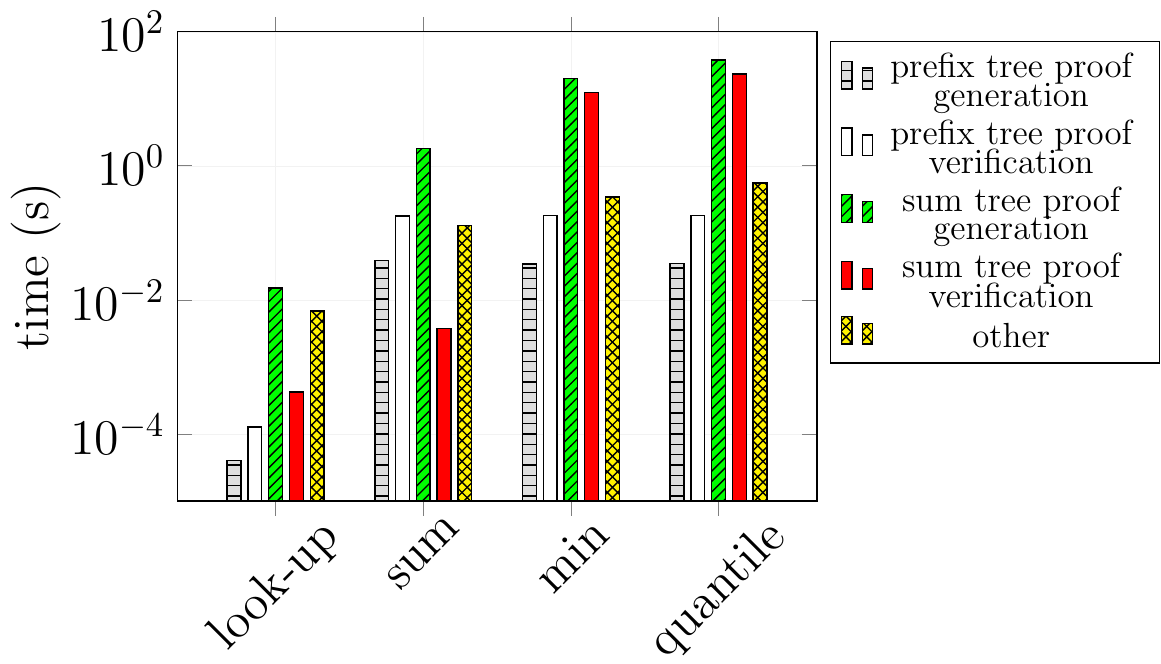}
\label{fig:time_breakdown}
}
\caption{(a) Proof sizes of different query types, (b) Time breakdown for different query types.}
\end{figure}

\subsection{End-to-End Performance}

We measure end-to-end latencies of different user queries in a setting with $n=1000$ and $m=5$. Figure~\ref{fig:time_breakdown} shows
the detailed breakdown. In particular, we divide the total time into five components:  
\begin{itemize}
\item Prefix tree proof generation: the time for generating the prefix tree inclusion proof for look-ups, and the
range cover proof 
for the other queries.
\item Prefix tree proof verification: the time for verifying the prefix tree proofs.
\item Sum tree proof generation: the time for generating the sum tree inclusion proof for look-ups, the hashes
needed to verify the commitments for sum queries, and the sum tree inclusion proof and range proofs for other queries.
\item Sum tree proof verification: the time for verifying the sum tree proofs. 
\item Other: network delay and any other costs.  
\end{itemize}

We observe that network delay is a significant factor in the end-to-end latency. As expected, this cost is proportional to the
proof sizes. The remaining cost is dominated by the cost of rebuilding the sum trees. For sum queries, the range cover
proofs have a large impact, but their costs are still an order of magnitude smaller than the sum tree proofs.  
The costs related to the sum tree proofs are especially large for min/max and quantile
queries because they are dominated by the cost of generating and verifying zero-knowledge proofs. The
end-to-end latencies for sum, min and quantile queries are $1s$, $23s$ and $60s$ respectively, which we believe to be reasonable for low-end virtual machines ({\em t2.xlarge} and {\em t2.micro}).

\subsection{Comparison Against Baselines}

\def\graphwidth{0.6}
\begin{figure*}[!t]
\centering
\makebox[\textwidth][c]{
\subfloat[][Server storage
cost.]{\includegraphics[scale=\graphwidth]{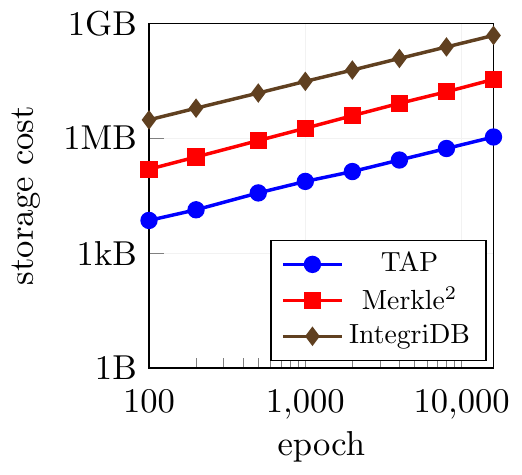}\label{fig:graph_storage}}
\subfloat[][Audit cost.]{\includegraphics[scale=\graphwidth]{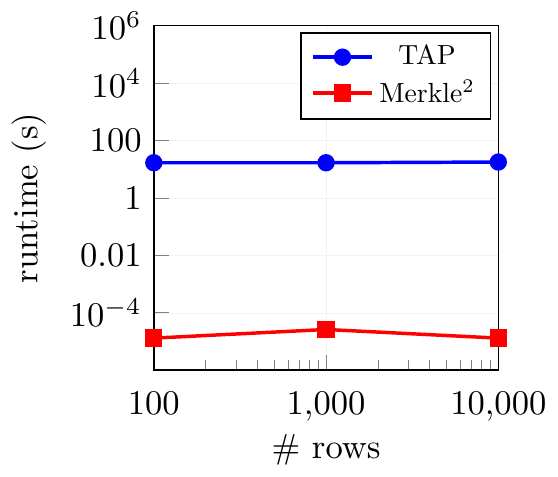}\label{fig:graph_auditor}}
\subfloat[][Insert query.]{\includegraphics[scale=\graphwidth]{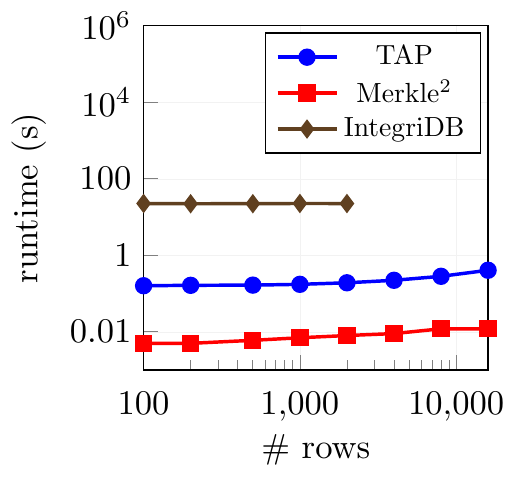}\label{fig:graph_insert}}
\subfloat[][Look-up query.]{\includegraphics[scale=\graphwidth]{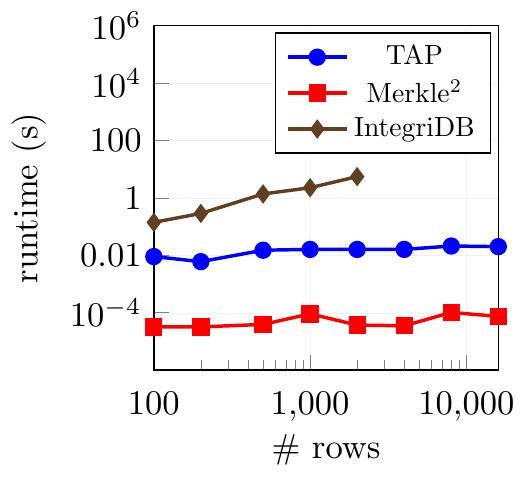}\label{fig:graph_lookup}}
\subfloat[][Sum query.]{\includegraphics[scale=\graphwidth]{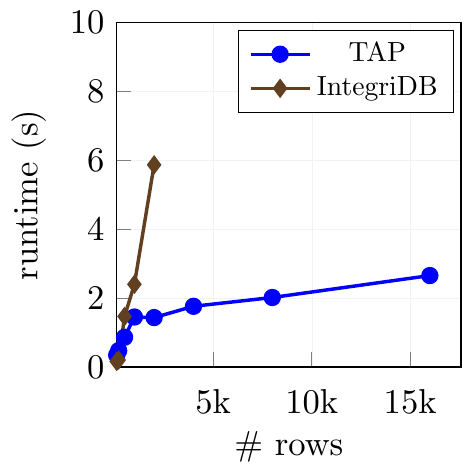}\label{fig:graph_sum}}
\subfloat[][Min/Max query.]{\includegraphics[scale=\graphwidth]{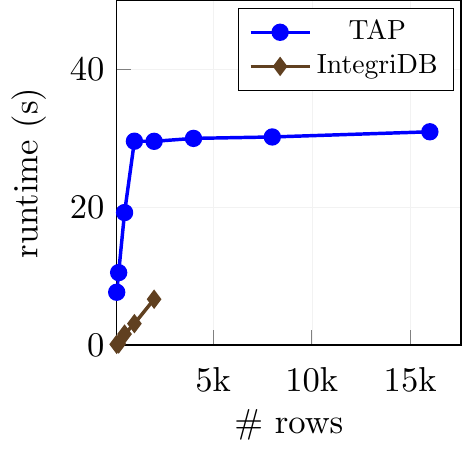}\label{fig:graph_min}}
}
\caption{Comparison between \name, Merkle$^2$, and IntegriDB in terms of storage, audit, and query costs.}
\end{figure*}

To provide empirical evidence for the asymptotic cost differences shown in \Cref{tab:comparison_asymptotic},
we compare \name against IntegriDB and Merkle$^2$ on different queries and with different data sizes.
Figure~\ref{fig:graph_storage} shows the storage costs for the three systems. Since \name does not store the Merkle
commitment trees but generates them on-the-fly during queries, it has the lowest storage cost.  IntegriDB has
the worst cost, as it needs to store at least
 25 copies of a tree with the same number of
leaves as Merkle$^2$.

In Figure~\ref{fig:graph_auditor}, we have displayed the audit costs of \name versus Merkle$^2$. We
omit IntegriDB from this graph as it does not support audits by default. It can be seen that the audits are
orders of magnitude more expensive in \name as in Merkle$^2$, which is due to the fact that an auditor in
\name has to evaluate dozens of zero-knowledge range proofs per epoch. However, we note that it takes only
$60s$ to audit an epoch with $100$ new entries on a low-end machine ({\em t2.xlarge}).  

\Cref{fig:graph_insert} displays the costs of inserting 100 data entries in an epoch. Merkle$^2$ is faster
than \name, as the former does not need to generate sum trees on-the-fly. However, IntegriDB has the worst
performance, because the new data needs to be inserted into at least 25 trees. The results
for look-up queries, shown in \Cref{fig:graph_lookup}, are similar to insertion queries, although the
  look-up costs for IntegriDB increase faster than the others. The IntegriDB reference implementation crashed when we performed
  queries on tables with more than 2000 rows.

Figure~\ref{fig:graph_sum} compares the cost of a sum query on the first 10 epochs for both \name and
IntegriDB. Although the proof size in IntegriDB depends only on the minimal covering sets of the values
that contribute to the sum, we observe that its overall cost appears linear in the number of rows.  We observe the
same for min queries, as shown in \Cref{fig:graph_min}. By extrapolating the line in \Cref{fig:graph_min}
beyond the point where the IntegriDB implementation crashed, we surmise that IntegriDB would be worse  than
\name before the table size reaches 10\,000 rows.

\subsection{Scalability}
\label{sec:scalability}

In this section, we investigate the performance of \name for realistic numbers of user IDs on a medium-end AWS machine ({\em m5.4xlarge}). \Cref{fig:scalability_tree} shows the cost of building \name's data structure for different numbers of IDs and sum trees. We observe a negligible difference between trees with 10 or 100 subtrees, but a noticeable difference between 100 and 1000 subtrees. In particular, building the tree takes roughly $50\%$ more time for 1000 subtrees than for 100 subtrees: the reason is that the construction of each subtree relies on a SQL select query to obtain the leaf values (i.e., $V_s$), which becomes a bottleneck when the number of subtrees is large. 

In \Cref{fig:scalability_audit}, we display the cost of a full audit as a function of the number of user IDs. We observe that for large numbers of IDs, the audit cost resembles a linear function of the number of IDs. For around $15\,000$ IDs, a full audit takes $420$ seconds, regardless of the number of subtrees.

\Cref{fig:scalability_quantile_all} displays the cost of a quantile query over the entire dataset. We see that the cost gradually becomes dependent only on the number of subtrees, as the main workload consists of creating zk-range proofs for each subtree. In \Cref{fig:scalability_quantile_limited}, we display the cost of querying a fixed range consisting of 10 subtrees. The cost of the query is independent of the total number of subtrees, and only has a logarithmic dependence on the number of IDs, which is invisible at realistic scales.

\begin{figure*}[t] \centering
\subfloat[][Tree construction.]{\includegraphics[scale=\graphwidth]{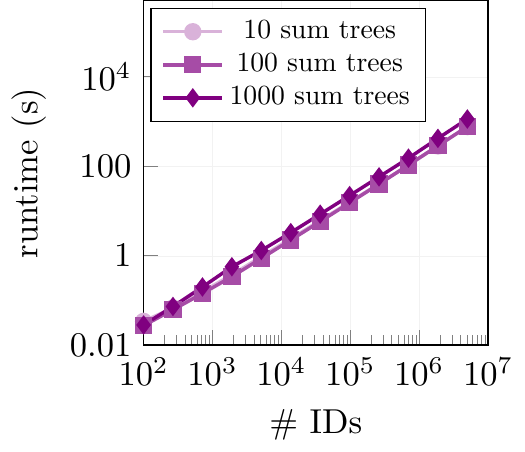}\label{fig:scalability_tree}}
\subfloat[][Audit.]{\includegraphics[scale=\graphwidth]{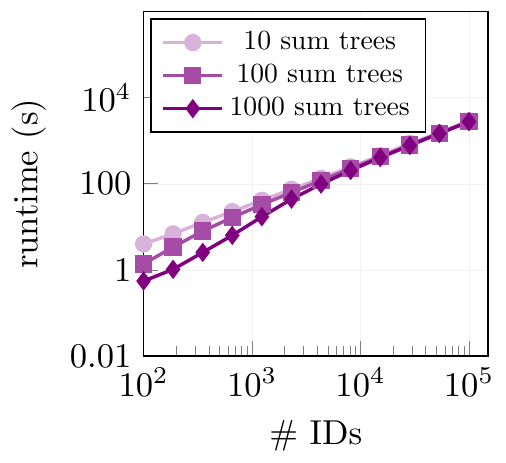}\label{fig:scalability_audit}}
\subfloat[][Quantile (full).]{\includegraphics[scale=\graphwidth]{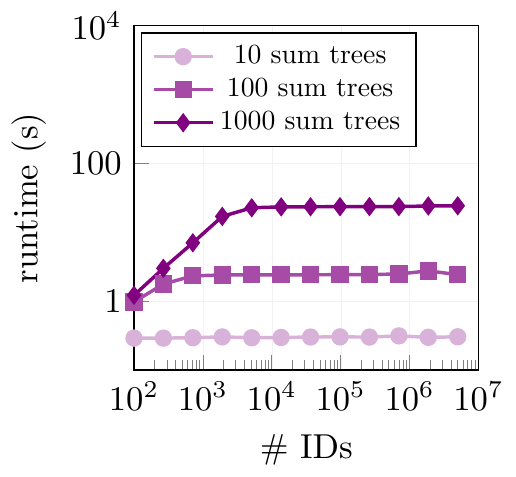}\label{fig:scalability_quantile_all}}
\subfloat[][Quantile (part).]{\includegraphics[scale=\graphwidth]{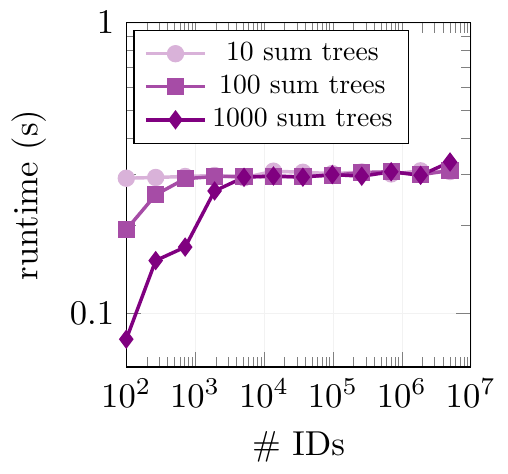}\label{fig:scalability_quantile_limited}}
\caption{Performance of \name for large numbers of user IDs.}
\end{figure*}

\section{Discussion \& Limitations}
\label{sec:discussion}

\textbf{Scaling Limitations.} From \Cref{fig:scalability_tree}, we observe that it takes around 285 seconds to generate \name's data structure for $\approx$1.8 million IDs and 100 subtrees. This is a reasonable overhead for the smart grid use case based on SP Group, which has around 1.6 million user IDs and 84 subtrees: if one tree has to be generated per hour, then even on an $\$1$/hour {\em m5.4xlarge} machine this only takes roughly $8\%$ of the available time. For the congestion pricing and digital advertising use cases, \name would not be able to support a system where every possible vehicle/website participates. However, for 250\,000 vehicles and 20 gantries, or for 500 websites and 10\,000 advertisers, the total number of user IDs would be 5 million, which a {\em m5.4xlarge} machine could process in 700 seconds.

With the same machine, an auditor can audit slightly over 100 user IDs per second. As such, the upper bound for feasibility for a full audit would be 360\,000 IDs for a system with 1-hour epochs. However, a randomized audit, in which a randomly chosen subset is audited, can be used for large systems to ensure that repeated misbehavior by the server is detected with high probability. Furthermore, targeted audits can be performed in cases where the operator profited from suspicious query results.

\textbf{Reporting Misbehavior.} The transparency property of \name ensures that if the server misbehaves in a verifiable way (e.g., signing a demonstrably false query result), then this can be detected and disseminated through gossip or the bulletin board. Cases in which the server misbehaves by falsifying user data are harder to detect and prove. In this case, we envisage that the user would appeal to a regulator or consumer watchdog. We assume that the server has a reputation to protect, which would be tarnished by frequent complaints.

\textbf{Incentives.} \name expects each user to independently and continuously verify her data and act as a whistleblower in the case of falsified data. As such, users in \name must have an incentive to monitor measurements. In the use cases of \Cref{sec:use_cases}, the users have a direct financial incentive to monitor their data because their bills or payments are directly impacted. 
From the perspective of the operator, transparency increases users' trust in the system and may help avoid frivolous lawsuits by being able to easily prove honest behavior.

\textbf{Trust Model.} In practice, a TAP user would not just rely on the operator for inserting data, but also for developing the app/client that verifies query proofs. In this setting, a malicious server could both falsify data \textit{and} deliberately insert bugs into the client to falsely convince users of the validity of the proofs. It is therefore vital that TAP users who require additional security have the ability to write their own code for proof verifications.

\textbf{User Identities.} TAP is impractical in settings where users are fully anonymous because such a setting would allow the server to create an unlimited number of fake users. However, we assume that a regulator would be able to detect the existence of (large numbers of) fake users, especially if they have a major impact on query results (e.g., a fake website that consistently achieves the highest click-through rate).

\section{Conclusions}
\label{sec:conclusions}
We have presented \name, a multi-user data service that provides data privacy, integrity, and transparency for user queries.  The ADS in \name combines a chronological prefix tree with sorted sum trees whose roots are stored in the prefix tree leaves. This data structure
allows \name to support a wide range of queries that are useful for emerging data services, with good
performance at scale. 

For future work, we plan to make the privacy costs to users explicit -- i.e., if an aggregate is taken
that includes a sum tree with very few leaves, then this reveals more information than if the sum trees each
have many leaves. 
This way, we can assign to
each user a privacy budget \cite{abadi2016deep}. 
Another interesting direction for future work is to add support for additional 
queries to \name, e.g., Spearman rank correlation between a recent set of measurements and another set of aggregates.

\section*{Acknowledgements}

This research / project is supported by the National Research Foundation, Singapore, under its National Satellite of Excellence Programme ``Design Science and Technology for Secure Critical Infrastructure'' (Award Number: NSoE\_DeST-SCI2019-0009), and by the National Research Foundation, Singapore, and Ministry of National Development, Singapore under its Cities of Tomorrow R\&D Programme (CoT Award COT-V2-2020-1). 
Tien Tuan Anh Dinh is supported by the Singapore University of Technology and Design startup grant, SRG ISTD 2019 144. Zheng Yang is supported by the Natural Science Foundation of China (Grant No. 61872051). Jianying Zhou is supported by A*STAR under its RIE2020 Advanced Manufacturing and Engineering (AME) Industry Alignment Fund - Pre Positioning (IAF-PP) Award A19D6a0053. 
Any opinions, findings and conclusions or recommendations expressed in this material are those of the author(s) and do not reflect the views of National Research Foundation, Singapore, Ministry of National Development, Singapore, or A*STAR.

\bibliographystyle{abbrv}
\bibliography{ref}

\appendix

\section{Pseudocode}
\label{sec:pseudocode}
This contains contains the pseudocode for three algorithms discussed in \Cref{sec:solution}. \Cref{alg:range_proof} determines a set of internal nodes in the prefix tree that covers the sum tree roots in the specified range, and is described under `Range Cover'. The range cover set acts as a proof of completeness for the given range, i.e., it demonstrates that no entries were incorrectly omitted or added by the server. It relies on another function, $\textnormal{RangeOverlap}(S^*,\node)$, which for a given interval node $\node$ in the prefix corresponds to a bit prefix that overlaps with the set $S^*$. (See also the function $\textnormal{DoesPrefixOverlapRange}$ in \mbox{trees/prefix\_tree.go} in the GitHub repository.) 

\Cref{alg:sum_count_server}  corresponds to the code that the server executes in response to a sum query, and is described under `Sum/Count/Average/Standard Deviation' in \Cref{sec:queries}. \Cref{alg:min_server} corresponds to the code that the server executes in response to a min query, and in described under `Min/Max' in \Cref{sec:queries}. The algorithm for quantiles is conceptually similar, and creates at most 2 range proofs for every sum tree.

\begin{algorithm}[h]
\caption{ Range proof (ExtendRangeProof)}
\label{alg:range_proof}
\footnotesize
\KwIn{node $\node$, node set $N$, range $S^*$}
\KwOut{node set $N'$}
\uIf{$\textnormal{RangeOverlap}(S^*,\node)$}{
	$N' = N \cup n$
	
	$L = \textsc{left}(\node)$
	
	$R = \textsc{right}(\node)$
	
	\uIf{$L \neq 0$}{
		$N' = N' \cup \textnormal{ExtendRangeProof}(L,N,S^*)$
	}
	\uIf{$R \neq 0$}{
		$N' = N' \cup \textnormal{ExtendRangeProof}(R,N,S^*)$
	}
	\Return $N'$
}
\Else{
  \Return $\emptyset$
 }
\end{algorithm}

\begin{algorithm}[h]
\caption{ Sum/Count query (Server)}
\label{alg:sum_count_server}
\footnotesize 
\KwIn{range $S^*$}
\KwOut{proof $\pi$}

$R \gets \textnormal{PrefixTree.GetRoot()}$

$N' \gets \textnormal{ExtendRangeProof}(R,\emptyset,S^*)$

$L'\gets \textnormal{GetLeaves}(N')$

$v \gets \textnormal{SQL.SumQuery}(S^*)$

\For{$j \in \{1,\ldots,z\}$}{
	$v^*_j \gets v^{j}$
}

$r^* \gets 0$

\For{$i \in L'$}{
	
	$r^* \gets r^* + \textnormal{GetTotalSeed}(\node)$

	$h'_\node \gets \textnormal{GetRootChildrenHash}(\node)$
	
	$c_{\node\,1},\ldots,c_{\node\,z} \gets \textnormal{GetRootCommitments}(\node)$
	
	$l_\node \gets \textnormal{GetRootNumLeaves}(\node)$
}

\Return $(N',L',r^*,v^*_1,\ldots,v^*_z,(h'_\node,c_{\node\,1},\ldots,c_{\node\,z},l_\node)_{\node \in L'})$

\end{algorithm} 

\begin{algorithm}[h]
\caption{ Min query (Server)}
\label{alg:min_server}
\footnotesize 
\KwIn{range $S^*$, large integer $K$}
\KwOut{proof $\pi$}

$R \gets \textnormal{PrefixTree.GetRoot()}$

$N' \gets \textnormal{ExtendRangeProof}(R,\emptyset,S^*)$

$L'\gets \textnormal{GetLeaves}(N')$

$v^* \gets \textnormal{SQL.MinQuery}(S^*)$

$P_{zk} \gets \textsc{NIZK.Setup}(1^{\kappa})$

$i^* \gets -1$

\For{$\node \in L'$}{
	
	$j^* \gets \textnormal{GetLeftmostLeaf}(\node)$
	
	$v \gets \textnormal{GetValue}(\node^*)$
	
	$c_{\node\,1},\ldots,c_{\node\,z} \gets \textnormal{GetLeafCommitments}(j^*)$
	
	$h'_\node \gets \textnormal{GetLeafUserTimeHash}(j^*)$
	
	$\pi'_\node \gets \textnormal{GenerateInclusionProof}(j^*,\node)$
	
	\uIf{$i^* = -1 \textbf{ and } v = v^*$}{
		$\pi^*_\node \gets \textsc{NIZK.Prove}(P_{zk}, v, [v^*, v^*+1))$
		
		$i^* \gets \node$
	} \Else {
		$\pi^*_\node \gets \textsc{NIZK.Prove}(P_{zk}, v, [v^*, K])$
	}
}
\Return $
(N',L',v^*,i^*,(c_{\node\,1},\ldots,c_{\node\,z}, h'_{\node}, \pi'_\node, \pi^*_\node)_{\node\in L'})$

\end{algorithm} 

\section{Security Model}
\label{sec:security_model}

\begin{table*}[htp]
\caption{Security games of TAP.} 
\label{tab:games}
\centering
\footnotesize
\begin{tabular}{|ll|}
\hline      
$\underline{\GameP^{\Transparancy}_{\TAP,\Adversary}(\kappa)}$                                                                                                                                                                                                           & $\underline{\GameP^{\Privacy}_{\TAP,\Adversary}(\kappa},\PrivacyType)$                                                                                                                                                \\
$\Qlist:=\emptyset$;                                                                                                                                                                                                                                                     & $\Qlist:=\emptyset$; $\CorruptList :=\emptyset$                                                                                                                                               \\
$(\ParmsTAP, k_s,\Delta_0) \leftarrow \TAPInitialize(1^\kappa)$                                                                                                                                                                                                          & $(\ParmsTAP, k_s,\Delta_0) := \TAPInitialize(1^\kappa)$                                                                                                                                       \\
$(i^*,t^*,\widetilde{\QueryType^*},\widetilde{\Delta_{t^*}},\widetilde{\TAPQueryRestult{t^*}},\widetilde{\TAPProof{t^*}},\widetilde{\TAPQueryMsg{i^*\,t^*}})\leftarrow \Adversary^{\OracleTAPQuery(\cdot,\cdot,\cdot,\cdot),\OracleInsertH(\cdot,\cdot)}(\ParmsTAP,k_s)$ & $(state, i^*,j^*,t^*,\TAPQueryMsg{0}^*,\TAPQueryMsg{1}^*)\leftarrow \Adversary^{\OracleTAPQuery(\cdot,\cdot,\cdot,\cdot),\OracleInsertH(\cdot,\cdot),\OracleCorrupt(\cdot,\cdot)}(\ParmsTAP)$ \\
$r_{i\,t} \leftarrow  \TAPEpochSecretGen(\ParmsTAP,k_s,i^*,t^*)$                                                                                                                                                                                                         & $\quad$, s.t. $t^*-1$ is the latest epoch, and all query messages have the same size,                                                                                                         \\
Return $\left((i^*,t^*,\widetilde{\QueryType^*},\widetilde{\Delta_{t^*}},\widetilde{\TAPQueryRestult{t^*}},\widetilde{\TAPProof{t^*}},\widetilde{\TAPQueryMsg{i^*\,t^*}}) \notin (\Qlist \cup \Hlist) \right)$                                                           & $b \rand \{0,1\}$                                                                                                                                                                             \\
$\quad \wedge\ \left((i^*,t^*) \in \Hlist \right)$ $\wedge\ \left(\TAPCheckEpoch(\ParmsTAP,\widetilde{\Delta_{t^*}})\right)$                                                                                                                                             & If $b=0$, then $\TAPQueryMsg{i^*\,t^*}:=\TAPQueryMsg{0}^*$ and $\TAPQueryMsg{j^*\,t^*+1}:=\TAPQueryMsg{1}^*$                                                                                  \\
$\quad \wedge\ \left(\TAPVerify(\ParmsTAP,r_{i^*\,t^*},\QueryType^*,\Delta_{t^*},\widetilde{\TAPQueryRestult{i^*\,t^*}},\widetilde{\TAPProof{i^*\,t^*}}\right)$                                                                                                          & Else $\TAPQueryMsg{i^*\,t^*}:=\TAPQueryMsg{1}^*$ and $\TAPQueryMsg{j^*\,t^*+1}:=\TAPQueryMsg{0}^*$                                                                                            \\
                                                                                                                                                                                                                                                                         & $(\Delta_{t^*},\TAPQueryRestult{i^*\,t^*},\TAPProof{i^*\,t^*}) \leftarrow \TAPQuery(\ParmsTAP,k_s,i^*,t^*,\Qinsert, \Delta_{t^*-1}, \TAPQueryMsg{i^*\,t^*})$                                  \\
$\underline{\OracleTAPQuery(i,t,\QueryType, \TAPQueryMsg{i\,t})}$:                                                                                                                                                                                                       & $(\Delta_{t^*+1},\TAPQueryRestult{j^*\,t^*+1},\TAPProof{j^*\,t^*+1}) \leftarrow \TAPQuery(\ParmsTAP,k_s,j^*,t^*,\Qinsert, \Delta_{t^*}, \TAPQueryMsg{j^*\,t^*+1})$                            \\
$(\Delta_t,\TAPQueryRestult{i\,t},\TAPProof{i\,t}) \leftarrow \TAPQuery(\ParmsTAP,k_s,i,t,\QueryType, \Delta_{t-1}, \TAPQueryMsg{i\,t})$                                                                                                                                 & $b' \leftarrow \Adversary^{\OracleTAPQuery(\cdot,\cdot,\cdot,\cdot),\OracleInsertH(\cdot,\cdot),\OracleCorrupt(\cdot,\cdot)}(\ParmsTAP,state,\Delta_{t^*},\Delta_{t^*+1})$                                                 \\
$(i,t,\QueryType,\Delta_t,\TAPQueryRestult{i\,t},\TAPProof{i\,t},\TAPQueryMsg{i\,t}) \rightarrow \Qlist$                                                                                                                                                                 & Return $(b=b')\ \wedge\ ((i^*,t^*) \notin \CorruptList)\ \wedge\ (j^*,t^*+1) \notin \CorruptList$                                                                                             \\
Return $\Delta_t,\TAPQueryRestult{i\,t},\TAPProof{i\,t}$                                                                                                                                                                                                                 &                                                                                                                                                                                               \\
                                                                                                                                                                                                                                                                         & $\underline{\OracleCorrupt(i,t)}$:                                                                                                                                                            \\
$\underline{\OracleInsertH(i,t)}$:                                                                                                                                                                                                                                       & $r_{i\,t} \leftarrow  \TAPEpochSecretGen(\ParmsTAP,k_s,i,t)$                                                                                                                                  \\
$\TAPQueryMsg{i\,t}\rand  \TAPMSpace$                                                                                                                                                                                                                                    & $(i,t,r_{i\,t}) \rightarrow \CorruptList $                                                                                                                                                    \\
$(\Delta_t,\TAPQueryRestult{i\,t},\TAPProof{i\,t}) \leftarrow \TAPQuery(\ParmsTAP,k_s,i,t,\Qinsert, \Delta_{t-1}, \TAPQueryMsg{i\,t})$                                                                                                                                   & Return $r_{i\,t}$                                                                                                                                                                             \\
$(i,t,\Qinsert,\Delta_t,\TAPQueryRestult{i\,t},\TAPProof{i\,t},\TAPQueryMsg{i\,t}) \rightarrow \Hlist$                                                                                                                                                                   &                                                                                                                                                                                               \\
Return $\Delta_t,\TAPQueryRestult{i\,t},\TAPProof{i\,t}$                                                                                                                                                                                                                 &                             
                                     \\ \hline
\end{tabular}
\end{table*}

Let $\kappa$ denote the security parameter, and $\emptyset$ the empty string. 
When $X$ is a set, $x \rand X$ denotes the action of sampling an element uniformly at random from $X$. %

\bparagraph{Syntax} 
We define a transparent and privacy-preserving data services (TAP) scheme using the following algorithms:

\hspace{0.2cm}$(\ParmsTAP, k_s,\Delta_0) \leftarrow \TAPInitialize(1^\kappa)$: run by the server. It takes as input the 
security parameter $\kappa$, and outputs system parameters $\ParmsTAP$, a secret key $k_s \in \TAPSKSpace$, and a mutable public verification state $\Delta_0 \in \TAPVTSpace$ of the TAP instance, where $\TAPSKSpace$ is a secret key space and $\TAPVTSpace$ is the space for public verification state. 

\hspace{0.2cm}$r_{i\,t} \leftarrow  \TAPEpochSecretGen(\ParmsTAP,k_s,i,t)$: run by the server. It takes as input the system parameters $\ParmsTAP$, the secret key $k_s$, and the epoch $t$, and outputs epoch secrets $r_{i\,t} \in  \TAPSSSpace$ for  user $i$, where $\TAPSSSpace$ is an epoch secret space. 

\hspace{0.2cm}$(\Delta_t,\TAPQueryRestult{i\,t},\TAPProof{i\,t}) \leftarrow \TAPQuery(\ParmsTAP,k_s,i, t,\QueryType, \Delta_{t-1}, \TAPQueryMsg{i\,t})$: run by server. It
takes as input the system parameters $\ParmsTAP$, secret key $k_s$,  a query type $$\QueryType \in \{\Qinsert,\Qlookup, \Qsum, \Qcount, \Qaverage,\Qminmax,\Qquantile\}$$ from a user $i$ for epoch $t$, the public verification state $\Delta_{t-1}$, and a query message $\TAPQueryMsg{i\,t} \in \TAPMSpace$, and outputs a query result $\TAPQueryRestult{i\,t} \in \TAPRSpace$ and the corresponding proof $\TAPProof{i\,t} \in 
\TAPProofSpace$, and an updated public verification state  $\Delta_t$ (if $\QueryType=\Qinsert$), where $\TAPProofSpace$ is a proof space $\TAPMSpace$ is the query message space, and $\TAPRSpace$ is a query result space. 

\hspace{0.2cm}$\{0,1\} \leftarrow \TAPVerify(\ParmsTAP,r_{i\,t},\QueryType,\Delta_t,\TAPQueryRestult{i\,t},\TAPProof{i\,t})$: run by a user $i$. It takes as input the system parameters $\ParmsTAP$, epoch secret $r_{i\,t}$, a query type $\QueryType$, the public verification state $\Delta_t$, and a query result $\TAPQueryRestult{i\,t}$ obtained from server for epoch $t$, and the corresponding proof $\TAPProof{i\,t}$, and outputs True (1) or False (0).

\hspace{0.2cm}$\{0,1\} \leftarrow\TAPCheckEpoch(\ParmsTAP,\Delta_t)$: run by auditor. It takes as input the system parameters $\ParmsTAP$ and the public verification state $\Delta_t$ for epoch $t$, and outputs True (1) if the verification state $\Delta_t$ is valid, and False (0) otherwise.

Given $(\ParmsTAP, k_s,\Delta_0) := \TAPInitialize(1^\kappa)$, any  user $i$'s epoch secret $r_{i\,t}$, any valid query message  $\TAPQueryMsg{i\,t} \in \TAPMSpace$ for a time
epoch $t$, and $(\Delta_t,\TAPQueryRestult{i\,t},\TAPProof{i\,t}) := \TAPQuery(\ParmsTAP,k_s,i,t,\QueryType, \Delta_{t-1}, \TAPQueryMsg{i\,t})$, we say that a TAP scheme is correct if  
$\TAPVerify(\ParmsTAP,r_{i\,t},\QueryType,\Delta_t,\TAPQueryRestult{i\,t},\TAPProof{i\,t}) = 1$ and $\TAPCheckEpoch(\ParmsTAP,\Delta_t) = 1$.

\bparagraph{Security Properties} 
Here, we define three properties of TAP: \emph{integrity}, \emph{transparency}, and \emph{privacy} (to achieve the requirements R2, R3a, and R3b). In Table~\ref{tab:games}, we formulate these properties via two games  $\GameP^{\Transparancy}_{\TAP,\Adversary}$ and $\GameP^{\Privacy}_{\TAP,\Adversary}$ 
 running between a challenger and an adversary $\Adversary$, respectively. We model both integrity and transparency of TAP in one game $\GameP^{\Transparancy}_{\TAP,\Adversary}$ for simplicity, since they share most of the  procedures and winning conditions in the game.

 In both games, the adversary is allowed to ask an oracle query $\OracleTAPQuery(i,t,\QueryType, \TAPQueryMsg{i\,t})$ to query any message $\TAPQueryMsg{i\,t}$ of her own choice for any query type $\QueryType$. Via this query, the adversary can either insert a malicious message into the data structure and also learn values based on compromised epoch secrets. Meanwhile, the epoch secrets can be compromised based on the oracle query  $\OracleTAPQuery(i,t,\QueryType,\TAPQueryMsg{i\,t})$. In addition, $\Adversary$ may ask the oracle query $\OracleInsertH(i,t)$ which is used to insert honest messages into the target TAP instance. By doing so, $\Adversary$ will try to break the transparency properties for some honest inserted messages. We model those security properties in a multiparty setting (to cover the requirement R1) because $\Adversary$ can ask those queries with an arbitrary user identity.

 In $\GameP^{\Transparancy}_{\TAP,\Adversary}$, the goal of $\Adversary$ is to generate a malicious message and the corresponding query results for an honest user $i^*$ and epoch $t^*$  (i.e., $(i^*,t^*,\widetilde{\QueryType^*},\widetilde{\Delta_{t^*}},\widetilde{\TAPQueryRestult{t^*}},\widetilde{\TAPProof{t^*}},\widetilde{\TAPQueryMsg{i^*\,t^*}})$) that are not generated by the challenger during the game but can pass the verification of either the honest  user $i^*$ or the auditor. To model privacy, the game $\GameP^{\Privacy}_{\TAP,\Adversary}$ is defined based on indistinguishability. Since the Min/Max and the quantile queries  would leak the concrete value of some user (without knowing its identity), we model the privacy by letting the adversary distinguish the owner of a malicious value (chosen by the adversary) from two honest parties for all kinds of queries. This approach can also cover the privacy of the (unleaked) value of a specific user as well. After all, if the adversary can know the value of a given honest user then she must be able to break the privacy formulated by  $\GameP^{\Privacy}_{\TAP,\Adversary}$ (i.e., distinguish the owner of the value). This leads to the following definition.

\begin{definition}
A transparent and privacy-preserving data services scheme $\TAP$ is secure if the advantages  
$$\Adv^{\Transparancy}_{\TAP,\Adversary}(\kappa)=\Pr[ \GameP^{\Transparancy}_{\TAP,\Adversary}(\kappa)=1]$$ and  $$\Adv^{\Privacy}_{\TAP,\Adversary}(\kappa)=\left|\Pr[ \GameP^{\Privacy}_{\TAP,\Adversary}(\kappa)=1] -1/2 \right|$$ of any PPT adversaries $\Adversary$ in the corresponding games are negligible. 
\label{def:security}
\end{definition}
\noindent Given $\GameP^{\Privacy}_{\TAP,\Adversary}$, the proof that TAP as discussed in \Cref{sec:solution} satisfies the properties of integrity, transparency, and integrity is similar to the proofs of Theorems 1--4 in \cite{esorics21}. Informally, \Cref{thm:limited_sum} can be proven using the property that an adversary cannot distinguish commitments from random \mbox{values}, and that a hidden value can be obtained from the (known) sum and the coalition's $f$ known values only if $n-f=1$. \Cref{thm:limited_quantile} follows from the knowledge that a quantile query reveals at most one value per query. 

\section{Differential Privacy}
\label{sec:differential_privacy}
In \Cref{sec:solution}, we have presented \name, a data service architecture in which the server returns accurate responses to user queries at the cost of revealing a limited number of user values. However, a stronger notion of privacy may be required in some contexts. For example, if the Value attribute corresponds to power usage, and a single residence is known to be the biggest power consumer in its neighborhood with high probability, then an adversarial user can learn this residence's exact power usage through a max query on this neighborhood.
\textit{Differential privacy}~\cite{dwork2006calibrating,shi2011privacy} provides a stronger notion of privacy for a data service. 
In differential privacy, data is obfuscated through the addition of random \textit{noise}.
We focus on an approach in which the server adds
noise to query results before returning them to the user~\cite{pinq}. 
An alternative design would be for the users to add random noise to their data according to a
fixed probability distribution, which would allow servers to compute aggregates without learning the data
of individual users~\cite{rappor}. However, the latter approach would prevent the server from obtaining data that may be necessary for, e.g., billing, so we focus on the former.

Improving privacy through the addition of noise necessarily reduces transparency: the noise values must be hidden (or else the obfuscated data can be reconstructed), so it is impossible to verify if they were generated in accordance with the agreed probability distributions. As such, we focus on a weaker notion of transparency, namely  that a malicious server is unable to \textit{arbitrarily} distort query results by manipulating the noise generation process.
As is common in the differential privacy literature, we focus on making the magnitude of the noise dependent on the \textit{sensitivity}, i.e., the maximum possible impact on a query result by removing a single user's value.
We can then utilize zero-knowledge range proofs to prove that the noise is within some range defined by the sensitivity. 
In the following, we focus on a single query and leave an investigation into the effects of multiple queries over the same range as future work.

Let $\data$ be the dataset, i.e., the set of values covered by the query's range. Let $\trueres(\data)$ be the query's \textit{true} result, and $\res(\data)$ the result that is returned by the server.
 Let $\RR$ be the result space, so that $\res(\data),\trueres(\data) \in \RR$. In our setting, $\RR$ is the space on which our commitments and zero-knowledge range proofs are defined. Let $\DD$ be the set of \textit{pairs} of datasets such that for any \mbox{$(\data,\datb) \in \DD$} a unique value $d \in \data$ exists such that $\data / \{d\} = \datb$, or $d \in \datb$ such that $\datb / \{d\} = \data$ -- i.e., those pairs that differ in exactly one value. In this setting, the protocol satisfies $(\epsilon,\delta)$-differential privacy if, for all $(\data,\datb) \in \DD$ and all $S \subset \RR$,
\begin{equation}
	\P(\res(\data) \in S) \leq e^\epsilon \cdot \P(\res(\datb) \in S) + \delta.
	\label{eq:differential_privacy}
\end{equation}
If $\delta = 0$, then the protocol satisfies ``pure'' $\epsilon$-differential privacy, whereas if $\delta>0$ it satisfies``approximate'' differential privacy. Let the \textit{sensitivity} of the query result $\res$ be defined as  
\begin{equation}
	\sens = \max_{(\data,\datb) \in \DD} | \trueres(\data) - \trueres(\datb) |.
	\label{eq:sensitivity}
\end{equation}
A general result for $\RR = \mathbb{R} $ \cite{dwork2006calibrating} states that returning a result $\res$ such that
	$$
	\res(\data) = \trueres(\data) + \noisea,
	$$
	where $\noisea$ is a Laplace-distributed random variable with scale parameter $\sigma$, guarantees $\epsilon'$-differential privacy with $\epsilon' = \sens /\sigma$. In particular, $\sigma = \sens$ guarantees $\epsilon$-differential privacy.

A challenge in our context is that the Laplace distribution has positive density on the entire interval $[-\infty,\infty]$. This is acceptable in cases where the threat model assumes that the server is always honest. However, in our case it would allow a malicious server to add \textit{arbitrarily large} noise to the true result, and therefore convince the user to accept any value desired by the malicious server. 
To limit the scope for server misbehavior, noise should instead be drawn from a bounded interval. 
Prior results in this area for truncated Gaussian \cite{liu2018generalized} and Laplace \cite{holohan2018bounded} noise indicate that although $\epsilon$-differential privacy cannot be achieved in this case, $(\epsilon,\delta)$-differential privacy is possible. 
In \cite{dagan2020bounded}, optimal asymptotic bounds were derived, but these are not practical in our setting as these bounds depend on hidden constants. 
As such, we focus in the following on a generic approach for noise on the bounded interval  $\{-\bound,-\bound+1,\ldots,\bound\}$ for a constant $\bound$.
\bparagraph{Theorem} In the following, let  $g : [0,\bound] \rightarrow [0,\infty)$ and \mbox{$G(z) = \sum_{x = 0}^z g(x) dx$} such that $g(0) +2G(\bound-1) = 1$. We then define the noise $\noisea = \res(\data) - \trueres(\data)$ as a random variable on $\{-\bound,-\bound+1,\ldots,\bound\}$ with the cumulative distribution function
\begin{equation}
\P(\noisea \leq z) = \left\{ \begin{array}{cl} G(z+\bound)  & \text{if } z \in [-\bound,0], \\ 1 - G(\bound-z+1) & \text{if } z \in [0,\bound]. \end{array} \right.
\label{eq:noisedist}
\end{equation}
Note that the probability distribution of $Z$ is symmetric, i.e., $\P(Z \leq -z) = \P(Z \geq z)$ for $z \in [0,\bound]$. 
Similarly, let $\noiseb = \res(\datb) - \trueres(\datb)$ be the noise under $\datb$, where $\P(\noiseb \leq z)$ is also as defined in \eqref{eq:noisedist}. We assume that $\res(\data) > \res(\datb)$ (this is w.l.o.g.\ because both $(\data,\datb)$ and $(\datb,\data)$ are part of $\DD$). We then have the following result.
\begin{theorem}
For the random variable $Z$ as defined in \eqref{eq:noisedist}, and for $\bound \geq \sens$ with $\sens$ as defined in \eqref{eq:sensitivity}, the mechanism $\res(\data) = \trueres(\data) +Z$ satisfies $(\epsilon,\delta)$-differential privacy with 
$$
\epsilon = \log\left(\max_{a \in A'} \frac{\P(\res(\data)=a)}{\P(\res(\datb)=a)}\right), \delta = G(\sens-1), \textit{ and}
$$
$$A' = \{a \in \mathbb{N}: a \geq \trueres(\data)-d, a \leq \trueres(\datb)+d\}.
$$
\end{theorem}
\begin{proof}

We first note that if $\bound \geq \frac{1}{2} \sens$, which is a condition of the theorem, then $\RR$ can be divided into five intervals: 
\begin{center}
\begin{tabular}{rcrl}
$A_1$ & $=$ & $(-\infty$,&$\trueres(\datb)-\bound)$, \\
$A_2$ & $=$ & $[\trueres(\datb)-\bound$,&$ \trueres(\data)-\bound)$, \\
$A_3$ & $=$ & $[\trueres(\data)-\bound$,&$ \trueres(\datb)+\bound]$, \\
$A_4$ & $=$ & $(\trueres(\datb)+\bound$,&$ \trueres(\data)+\bound]$, \\
$A_5$ & $=$ & $(\trueres(\data)+\bound$,&$ \infty)$,
\end{tabular}
\end{center}
which have the following properties:
\begin{center}
\begin{tabular}{cc}
$\forall S \subset A_1$, & $\P(\res(\data) \in S) = 0$ \;and\; $\P(\res(\datb) \in S) = 0$, \\
$\forall S \subset A_2$, & $\P(\res(\data) \in S) = 0$ \;and\; $\P(\res(\datb) \in S) > 0$, \\
$\forall S \subset A_3$, & $\P(\res(\data) \in S) > 0$ \;and\; $\P(\res(\datb) \in S) > 0$, \\
$\forall S \subset A_4$, & $\P(\res(\data) \in S) > 0$ \;and\; $\P(\res(\datb) \in S) = 0$, \\
$\forall S \subset A_5$, & $\P(\res(\data) \in S) = 0$ \;and\; $\P(\res(\datb) \in S) = 0$. \\
\end{tabular}
\end{center}
We first prove the main theorem for any set $S$ that is a subset of either of these five sets. For $S \subset A_1 \cup A_5$ we trivially have that \eqref{eq:differential_privacy} holds for all $\epsilon \in \mathbb{R}$ and $\delta=0$ because 
$$
\P(R(D) \in S) = \P(R(D') \in S) = 0.
$$
For $S \subset A_2$, we have that $\P(\res(\data) \in S) = 0$ and \mbox{$\P(\res(\datb) \in S) > 0$}. We find that 
\begin{equation}
\begin{split}
\P(\res(\datb) \in S)
& \leq \P(\res(\datb) \in A_2) \\ 
& = \P(\noiseb < \trueres(\data) - \trueres(\datb) - \bound) - \P(\noiseb < -\bound) \\
& \leq \P(\noiseb < \sens - \bound), \\
\end{split}
\end{equation}
 where the last inequality holds because of the definition of the sensitivity and because $\P(\noiseb < -\bound) = 0$.  As we assume $\bound > \sens$, $\sens - \bound < 0$ and $\P(\noiseb < \sens - \bound) = G(\sens-1)$, and in this case we have $(\epsilon,\delta)$-differential privacy with $\epsilon=0$ and $\delta = G(\sens-1)$.

For $S \subset A_4$, we have that $\P(\res(\data) \in S) > 0$ and \mbox{$\P(\res(\datb) \in S) = 0$}. We find that 
\begin{equation}
\begin{split}
\P(\res(\data) \in S)
& \leq \P(\res(\data) \in A_4) \\ 
& = \P(Z \leq \bound) - \P(\noisea \leq \trueres(\datb)- \trueres(\data)+\bound) \\
& = \P(\noisea > \trueres(\datb)- \trueres(\data)+\bound) \\
& = \P(\noisea < \trueres(\data)-\trueres(\datb)-\bound) \\
& \leq \P(Z < \sens-\bound) \\
\end{split}
\end{equation}
In the above, we use $\P(Z \leq \bound) = 1$, $\P(Z \leq \bound) = 1 - \P(Z > \bound)$, symmetry of the distribution of $Z$, and $\trueres(\datb) - \trueres(\data)+\bound \geq \bound-\Delta \geq 0$. Again, we have $(\epsilon,\delta)$-differential privacy with $\epsilon=0$ and $\delta = G(\sens-1)$.

For $S \subset A_3$, we have that $\P(\res(\data) \in S) > 0$ and \mbox{$\P(\res(\datb) \in S) > 0$}. Let 
$$
\pp = \max_{a \in A_3} \frac{\P(\res(\data)=a)}{\P(\res(\datb)=a)}
$$
Then
$$
\P(\res(\data)=a) = \frac{\P(\res(\data)=a)}{\P(\res(\datb)=a)} \P(\res(\datb)=a) \leq \pp \P(\res(\datb)=a)
$$
which means that we have $(\epsilon,\delta)$-differential privacy with $\epsilon = \log(\pp)$ for any $S = \{a\}$, which extends to any set $S \subset A_3$ by a union bound.

Finally, having shown $(\epsilon,\delta)$-differential privacy for any $S$ that is a subset of either of $A_1,\ldots,A_5$, we can use a union bound to establish $(\epsilon,\delta)$-differential privacy for any union of those subsets. This completes the proof of the theorem.

\end{proof}

\end{document}